\documentclass{article}
\usepackage{amssymb}
\usepackage{titlesec}
\usepackage{graphicx}
\usepackage[hidelinks]{hyperref}
\usepackage{float}
\usepackage{mathtools}
\usepackage{afterpage}
\usepackage{amsmath}
\usepackage[export]{adjustbox}
\usepackage{graphicx}
\usepackage{subcaption}
\usepackage{dsfont}
\usepackage{graphicx}
\usepackage[ruled,vlined]{algorithm2e}
\usepackage[title]{appendix}
\usepackage{geometry}
\usepackage{enumitem}
\DeclarePairedDelimiter{\floor}{\lfloor}{\rfloor}

\usepackage{color}
\usepackage[english]{babel}
\titleformat{\chapter}[display]
  {\normalfont\LARGE\bfseries}
  {\titleline{}\vspace{5pt}\titleline{}\vspace{1pt}%
  \MakeUppercase{\chaptertitlename} \thechapter}
  {1pc}
  {\titleline{}\vspace{0.5pc}} 
\usepackage{mathtools}
\DeclarePairedDelimiter\abs{\lvert}{\rvert}

\makeatletter

\renewcommand\section{\@startsection {section}{1}{\z@}%
                               {-3.5ex \@plus -1ex \@minus -.2ex}%
                               {2.3ex \@plus.2ex}%
                               {\normalfont\large\bfseries}}
\renewcommand\subsection{\@startsection{subsection}{2}{\z@}%
                                 {-3.25ex\@plus -1ex \@minus -.2ex}%
                                 {1.5ex \@plus .2ex}%
                                 {\normalfont\bfseries}}

\newcommand{\distas}[1]{\mathbin{\overset{#1}{\kern\z@\sim}}}%
\newsavebox{\mybox}\newsavebox{\mysim}
\newcommand{\distras}[1]{%
  \savebox{\mybox}{\hbox{\kern3pt$\scriptstyle#1$\kern3pt}}%
  \savebox{\mysim}{\hbox{$\sim$}}%
  \mathbin{\overset{#1}{\kern\z@\resizebox{\wd\mybox}{\ht\mysim}{$\sim$}}}%
}
\makeatother

\usepackage[english]{babel} 
\usepackage{amsmath,amsfonts,amsthm} 

\usepackage{bold-extra}

\usepackage{fancyhdr}        
\numberwithin{equation}{section}       
\numberwithin{figure}{section}         
\numberwithin{table}{section}          

\newtheorem{theorem}{Theorem}[section]
\newtheorem{prop}[theorem]{Proposition}
\newtheorem{corollary}[theorem]{Corollary}
\newtheorem{lemma}[theorem]{Lemma}
\theoremstyle{remark}
\newtheorem{remark}{Remark}[section]
\theoremstyle{definition}

\newtheorem{assumption}{Assumption}
\theoremstyle{definition}
\newtheorem{definition}[theorem]{Definition}

\graphicspath{{img/}}

\title{
\normalfont \normalsize   
\huge      Optimal friction matrix for underdamped Langevin sampling      
}
\author{ Martin Chak \and Nikolas Kantas \and Tony Leli\`evre \and Grigorios A. Pavliotis}

\date{\normalsize\today}

\begin{document}

\maketitle

\begin{abstract}
A systematic procedure for optimising the friction coefficient in underdamped Langevin dynamics as a sampling tool is given by taking the gradient of the associated asymptotic variance with respect to friction. We give an expression for this gradient in terms of the solution to an appropriate Poisson equation and show that it can be approximated by short simulations of the associated first variation/tangent process under concavity assumptions on the log density. Our algorithm is applied to the estimation of posterior means in Bayesian inference problems and reduced variance is demonstrated when compared to the original underdamped and overdamped Langevin dynamics in both full and stochastic gradient cases.
\end{abstract}
\tableofcontents
\section{Introduction}
Let $\pi$ be a probability measure on $\mathbb{R}^n$ with smooth positive bounded density, also denoted $\pi$, with respect to the Lebesgue measure on $\mathbb{R}^n$ and let $f\in L^2(\pi)$ be an observable. In a range of applications including molecular dynamics \cite{MR2583309,MR3463433,MR3509213} and machine learning \cite{NIPS1992_f29c21d4,MR3555050,10.5555/3104482.3104568}, a quantity of interest is the expectation of $f$ with respect to $\pi$,
\begin{equation*}
\pi(f) := \int f d\pi,
\end{equation*}
which is analytically intractable and is numerically approximated most commonly by Markov Chain Monte Carlo (MCMC) methods, whereby $\pi$ is sampled by simulating an ergodic Markov chain $(X_k)_{1\leq k \leq N}$ with $\pi$ as its unique invariant measure and $\pi(f)$ is approximated by the empirical average $\frac{1}{N}\sum_{k = 1}^N f(X_k)$. MCMC methods enjoy central limit theorems for many Markov chains employed, the most well-known (class) of such methods being the Metropolis-Hastings algorithm \cite{MR3363437, doi:10.1063/1.1699114}. Recent efforts have been to develop MCMC methods suited to settings where $n\gg 1$ and where point evaluations of $\pi$ or its gradients are computationally expensive; these methods include slice sampling \cite{dubois2014approximate, MR1994729}, Hamiltonian Monte Carlo \cite{MR3648031, MR3960671, MR2858447}, piecewise-deterministic Markov processes \cite{MR3911113, MR3832232, vanetti2018piecewisedeterministic} and those based on discretisations of continuous-time stochastic dynamics \cite{MR4048994, MR3509213, MR4278799} together with divide-and-conquer and subsampling approaches \cite{MR3670492}.

In this paper we consider the underdamped Langevin dynamics. Denoting $\mathbb{S}_{++}^n$ as the set of real symmetric $n\times n$ positive definite matrices, 
the underdamped Langevin dynamics\footnote{also referred to as Langevin, second-order Langevin or kinetic Langevin.} with mass $M\in\mathbb{S}_{++}^n$ and friction matrix $\Gamma\in \mathbb{S}_{++}^n$ is given by the $\mathbb{R}^{2n}$-valued solution $(q_t,p_t)$ to
\begin{subequations}\label{langevin00}
\begin{align}
dq_t &= M^{-1}p_tdt\\
dp_t &= -\nabla U(q_t) -\Gamma M^{-1} p_tdt + \sqrt{2\Gamma} dW_t,
\end{align}
\end{subequations}
where $\sqrt{\Gamma}\in\mathbb{R}^{n\times n}$ is any matrix satisfying
\begin{equation*}
\sqrt{\Gamma}\sqrt{\Gamma}^\top = \Gamma,
\end{equation*}
$U:\mathbb{R}^n\rightarrow\mathbb{R}$ is the associated smooth potential or negative log density such that $\pi \propto e^{-U}$ and $W_t$ denotes a standard Wiener process on $\mathbb{R}^n$. 
The probability distribution from underdamped Langevin dynamics converges under general assumptions to the invariant probability measure given by
\begin{equation}\label{pitildeden}
\tilde{\pi}(dq,dp) = Z^{-1}e^{-U(q) - \frac{p^\top M^{-1} p}{2}}dqdp
\end{equation}
for a normalising constant $Z$ and there have been numerous recent works \cite{pmlr-v75-cheng18a, MR4091098, durmus2021uniform, Foster2021TheSO, he2020ergodicity, MR3040887, MR4309974, SanzSerna2021WassersteinDE} on its discretisations in terms of the quality of convergence to $\tilde{\pi}$ over time measured by (e.g.) Wasserstein distance; in this paper, 
the goal is to optimise $\Gamma\in\mathbb{S}_{++}^n$ directly with respect to the asymptotic variance in the convergence of 
\begin{equation*}
\pi_T(f) := \frac{1}{T}\int_0^T f(q_t)dt
\end{equation*}
to $\pi(f)$ for any particular $f$ (or a finite set of observables) as $T\rightarrow\infty$.\\ 
We mention that parameter tuning in MCMC methods is a widely considered topic \cite{MR2461882, MR4140028} (and references within).
Specifically for underdamped Langevin dynamics, tuning the momentum part of $\tilde{\pi}$ with respect to reducing metastability or computational effort was considered in \cite{MR3529154, MR3799045, MR3622623}. The choice of friction (as a scalar) has been a subject of consideration as early as in \cite{HOROWITZ1987510}, then in \cite{MR3774648, PhysRevE.75.056707, doi:10.1080/08927022.2020.1791858, skeel2021choice} within the context of molecular dynamics and also in \cite{MR4091098, MR3723428}. Most of these works make use of different measures for efficiency. The present work constitutes the first systematic gradient procedure for choosing the friction matrix in an optimal manner, with respect to a appropriate cost criterion. 
\subsection{Outline of approach}
We proceed with a formal description of our approach, precise statements can be found in the main Theorems~\ref{funcdereq} and~\ref{newformprop}. It is known using results from \cite{e19120647} and \cite{MR663900} that, under suitable assumptions on $U$ and $f$, a central limit theorem
\begin{equation}\label{est2}
\frac{1}{\sqrt{T}}\int_0^T (f(q_t)-\pi(f))dt \overset{\mathcal{D}}{\rightarrow} \mathcal{N}(0,\sigma^2)\qquad \textrm{as } T\rightarrow \infty
\end{equation}
holds and that $\sigma^2$, the asymptotic variance, has the form
\begin{equation}\label{av0}
\sigma^2 = 2\int\phi(f-\pi(f))d\tilde{\pi}
\end{equation}
where $\phi$ is a solution to the Poisson equation
\begin{equation}\label{poisson0}
-L\phi = f- \pi(f)
\end{equation}
and $L$ denotes the infinitesimal generator associated to~\eqref{langevin00}. 
Two key observations are then made. Firstly, for any direction $\delta\Gamma\in\mathbb{R}^{n\times n}$ in the friction matrix, the derivative of $\sigma^2$ with respect to the entries of $\Gamma$ in the direction $\delta\Gamma$, denoted $d\sigma^2.\delta\Gamma$, is given by the formula  
\begin{equation}\label{funcder00}
d\sigma^2.\delta\Gamma = -2\int (\nabla\!_p \phi)^\top \delta\Gamma \nabla\!_p \tilde{\phi} d\tilde{\pi},
\end{equation}
where $\tilde{\phi}$ is given by
\begin{equation}\label{phitilde}
\tilde{\phi}(q,p) = \phi(q,-p).
\end{equation}
A direction $\delta\Gamma$ that guarantees a decrease in $\sigma^2$ is then
\begin{equation}\label{DeltaGam}
\Delta\Gamma := \int \nabla\!_p \phi \otimes \nabla\!_p\tilde{\phi} d\tilde{\pi}
\end{equation}
where $\otimes$ denotes the outer product. Similarly, taking $\delta\Gamma$ to be the diagonal elements of~\eqref{DeltaGam} or $\delta\Gamma = I_n \int \nabla\!_p \phi \cdot \nabla\!_p \tilde{\phi} d\tilde{\pi}$ 
give in both cases a negative change in asymptotic variance respectively for diagonal $\Gamma$ and $\Gamma$ of the form $c I_n$. The second observation is that since the solution $\phi$ to the Poisson equation~\eqref{poisson0} is known to be given by
\begin{equation}\label{pform}
\phi(q,p) = \int_0^\infty \mathbb{E}[f(q_t)]dt,
\end{equation}
where $(q_t,p_t)$ solves~\eqref{langevin00} with initial condition $(q_0,p_0) = (q,p)$, given convexity conditions on the potential $U$ and under suitable assumptions, we have 
\begin{equation}\label{nabphi0}
\nabla\!_p \phi = \int_0^\infty \mathbb{E}[\nabla f(q_t)^\top D_p q_t]dt,
\end{equation}
where $D_p q_t$ denotes the $\mathbb{R}^{n\times n}$-matrix made of partial derivatives of $q_t$ with respect to the initial condition $p$ in momentum. Not only does $D_p q_t$ satisfy the dynamics that result from taking partial derivatives in~\eqref{langevin00}, which are susceptible to algorithmic simulation, but the process also decays to zero exponentially quickly, so that the infinite time integral~\eqref{nabphi0} can be accurately approximated with a truncation using short simulations of $D_p q_t$ for adaptive estimations of the direction~\eqref{DeltaGam} in $\Gamma$. This leads to an adaptive algorithm involving the selection of $\Gamma$ in an appropriate constrained set, of which we illustrate the performance with numerical examples.\\


Examples where improved $\Gamma$ can be found analytically are presented in Section~\ref{anacas}. 
Numerical illustrations making use of~\eqref{funcder00} and~\eqref{nabphi0} are presented in Sections~\ref{conc}. In particular, the algorithm is applied on the problem of finding the posterior mean in a Bayesian logistic regression inference problem for two datasets with hundreds of dimensions, where the best friction matrices found in both cases are close to zero (for example $\Gamma = 0.1I_n$ performs well compared to $\Gamma = I_n$, demonstrating reduced variance of almost an order of magnitude in Tables~\ref{table1} and~\ref{table2}).\\ 
To use the asymptotic variance for a particular observable (or a set of them) and to use measures for the quality of convergence to $\tilde{\pi}$ or to minimise an autocorrelation time as considered in \cite{MR3774648, PhysRevE.75.056707, MR4091098, HOROWITZ1987510, doi:10.1080/08927022.2020.1791858, skeel2021choice} can be conflicting goals. 
To elaborate, in \cite{HOROWITZ1987510}, the autocorrelation time was used as the point of comparison in the Gaussian target measure case for the optimal friction. For $n=1$, $\omega,\gamma\in\mathbb{R}$, $U(q) = \frac{1}{2}\omega^2 q^2$, $M = 1$, $\Gamma = \gamma>0$, the autocorrelation time for~\eqref{langevin00} satisfies 
\begin{equation}\label{autodec}
\partial_t \begin{pmatrix}\mathbb{E}(q_t q_0) \\ \mathbb{E}(p_t q_0)\end{pmatrix} = \begin{pmatrix} 0 & 1\\ -\omega^2 & -\gamma \end{pmatrix} \begin{pmatrix} \mathbb{E}(q_t q_0) \\ \mathbb{E}(p_t q_0)\end{pmatrix}.
\end{equation}
By considering the eigenvalues, the conclusion in \cite{HOROWITZ1991247} is that the optimal $\gamma$ for minimising the magnitude of $\mathbb{E}(q_t q_0)$ is given by the critical damping $\gamma = 2\omega$. A similar conclusion can be made when considering the spectral gap\cite{MR3288096}.
\begin{figure}[h]
  \centering
    \includegraphics[width=0.5\linewidth]{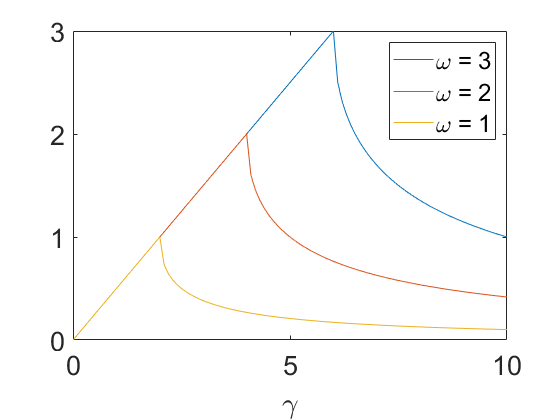}
\caption{The values $\min_i(\abs*{\textrm{Re}(\lambda_i)})$, where $\lambda_i$ are the eigenvalues of the matrix appearing in~\eqref{autodec}, also the spectral gap. Critical values of $\gamma$ are given by $2\omega$.}\label{specga}
  \end{figure}
On the other hand, if $f(q)=q$ in our setting, formally, the quantity $\int \int_0^\infty \mathbb{E}(q_t q_0) dt d\pi(q_0)$ is the asymptotic variance due to~\eqref{av0} and~\eqref{pform}. Despite the similarity, Corollary~\ref{lincor} asserts that $\gamma = 0$ is optimal.  
A more detailed discussion about Corollary~\ref{lincor} is given in Section~\ref{oddpol}. This difference emphasizes that, at the cost of generic convergence to $\tilde{\pi}$, the tuning of $\Gamma$ here is directed at variance reduction for a particular observable, in this case $f(q) = q$. 
However, multiple asymptotic variances can be used for the objective function to minimise, so that $\Gamma$ can be optimised with respect to several observables of interest simultaneously. Remark~\ref{simul} describes the implementation for a linear combination of asymptotic variances at no extra cost in terms of evaluations of $\pi$ or its gradients.

The rest of the paper is organised as follows. In Section~\ref{setting}, we provide a mathematical setting in which the underdamped Langevin dynamics with a friction matrix and in particular~\eqref{langevin00} has a well-defined solution and satisfies the central limit theorem for suitable observables, together with notations used throughout the paper. In Section~\ref{fdsec}, prerequisite results and the main formulae~\eqref{funcder00} and~\eqref{nabphi0} are precisely stated. 
Exact results concerning improvements in $\Gamma$ including the quadratic $U$, quadratic $f$ and linear $f$ cases are given in Section~\ref{anacas}. Numerical methods in approximating~\eqref{DeltaGam} together with an algorithm resulting from~\eqref{funcder00} and~\eqref{nabphi0} is outlined and detailed in Algorithm~\ref{algo0} and~\ref{algorithm} respectively in Section~\ref{nummeth}, alongside examples of $U$ and $f$ where improvements in variance are observed. In Section~\ref{proofs}, deferred proofs are given. In Section~\ref{discussion}, we conclude and discuss about future work.

\section{Setting}\label{setting}
Let $(\Omega,\mathcal{F},\mathbb{P})$ be a complete probability space, $(\mathcal{F}_t)_{t\in\mathbb{R}}$ be a normal (satisfying the usual conditions) filtration with $(W_t)_{t\geq 0}$ a standard Wiener process on $\mathbb{R}^n$ with respect to $(\mathcal{F}_t)_{t\in\mathbb{R}}$, $\tilde{\pi}$ be a probability measure given by~\eqref{pitildeden}. 

\begin{assumption}\label{smu}\textcolor{white}{a}
$U\in C^\infty(\mathbb{R}^n)$ satisfies 
$U\geq 0$ and its second derivatives satisfy \begin{equation}\label{assump2eq0}
\| D^2 U \|_\infty := \sum_{i,j} \sup_{q\in\mathbb{R}^{2n}}\abs*{\partial_{q_i}\partial_{q_j} U(q)} < \infty.
\end{equation}
\end{assumption}Note that~\eqref{assump2eq0} implies
\begin{equation}\label{assump2eq}
\abs*{\nabla U(q)} \leq K_U(1+\abs*{q})
\end{equation}
for some $K_U >0$. 
The existence and uniqueness of a strong solution to~\eqref{langevin00} is established in Theorem~\ref{dent}. Due to the smoothness of $U$ and $\Gamma$, the coefficients in~\eqref{langevin00} are locally Lipschitz and well-posedness of equation~\eqref{langevin00} is given by \cite{MR2329435}, to which we also refer to for the sense of solution. In addition, we make certain to satisfy the joint measurability assumption in \cite{MR663900} of~\eqref{trans}.

\subsection{Preliminaries and notation}
The set of smooth compactly supported functions is denoted $C_c^\infty$. The infinitesimal generator $L$ (defined in~\eqref{l2gen}) associated to~\eqref{langevin00} is given formally by its differential operator form, denoted $\mathcal{L}$, when acting on the subset $C_c^\infty(\mathbb{R}^{2n})$,
\begin{equation}\label{infgen}
\mathcal{L} = p^\top M^{-1}\nabla\!_q - \nabla U(q) ^\top \nabla\!_p - p^\top M^{-1}\Gamma \nabla\!_p + \nabla\!_p^\top \Gamma \nabla\!_p.
\end{equation}
Its formal $L^2$-adjoint $L^\top$ satisfies
\begin{equation}\label{inv}
L^\top \tilde{\pi} = 0,
\end{equation}
so that $\tilde{\pi}$ (see~\eqref{pitildeden}) is an invariant probability measure for~\eqref{langevin00} for a normalisation constant $Z$. Let 
\begin{equation*}
L_0^2(\pi) := \{ g\in L^2(\pi) : \int g d\pi = 0 \}
\end{equation*}
and similar for $\tilde{\pi}$. The notation $D^2 U$ will be used for the Hessian matrix of $U$. As in the introduction, $I_n\in\mathbb{R}^{n\times n}$ denotes the identity matrix. For matrices $A$, $\abs*{A}$ denotes the operator norm associated with the Euclidean norm on $\mathbb{R}^n$. $e_i$ is used to denote the $i^{\textrm{th}}$ Euclidean basis vector. For $A, B\in\mathbb{R}^{n\times n}$, $A:B := \sum_{i,j}A_{ij}B_{ij}$ and $A_S = \frac{1}{2}(A+A^\top)$. $\langle\cdot,\cdot\rangle_{\tilde{\pi}}$ denotes the inner product in $L^2(\tilde{\pi})$ and similar for $\pi$.

\subsection{Semigroup bound, Poisson equation and central limit theorem}
In this section, a central limit theorem for the solution to~\eqref{langevin00} is established, where the resulting asymptotic variance will be used as a cost function to optimise $\Gamma$ with respect to. Specifically, it will be shown that under some weighted $L^\infty$ bound on the observable $f\in L^2(\pi)$, the estimator $\pi_T$ for the unique solution $(q_t,p_t)$ to~\eqref{langevin00} converges to $\pi(f)$ as $T\rightarrow \infty$ such that~\eqref{est2} holds with~\eqref{av0}.\\[1em]
It is well known that the asymptotic variance can be expressed in terms of the solution to the Poisson equation~\eqref{poisson0} using the Kipnis-Varadhan framework, 
see for example Chapter 2 in \cite{MR2952852}, Section 3.1.3 in \cite{MR3509213}, \cite{MR3069369} and references therein. 
In order to show that the expression~\eqref{pform} is indeed a solution to the Poisson equation~\eqref{poisson0}, exponential decay of the semigroup~\eqref{semigroup} is used. In Theorem~\ref{thmc} below, we establish convergence in law to the invariant measure for the Langevin dynamics~\eqref{langevin00}.
For this, let the Lyapunov function $\mathcal{K}_l:\mathbb{R}^{2n}\rightarrow\mathbb{R}$ for all $l\in\mathbb{N}$ be given by
\begin{equation}\label{lyafunc}
\mathcal{K}(z) = \mathcal{K}_l(q,p) = \Big(cU(q) + a\abs*{q}^2 + b\langle q,p\rangle + \frac{c}{2}\abs*{p}^2 + 1\Big)^l
\end{equation}
for constants $a,b,c>0$.
\begin{assumption}\label{assump1}
There exist constants $\beta_1,\beta_2>0$ and $\alpha\in\mathbb{R}$ such that
\begin{equation}\label{assumpgrad}
\forall q\in\mathbb{R}^n, \langle q, \nabla\!_q U(q) \rangle \geq \beta_1 U(q) + \beta_2 \abs*{q}^2 + \alpha.
\end{equation}
\end{assumption}
Inequality~\eqref{assumpgrad} implies
\begin{equation}\label{Uquad}
U(q) \geq C\abs*{q}^2 - C
\end{equation}
for all $q\in\mathbb{R}^{n}$ and some generic constant $C>0$.
\begin{theorem}\label{thmc}
Under Assumptions~\ref{smu} and~\ref{assump1}, 
$\tilde{\pi}$ is the unique invariant probability measure for the SDE~\eqref{langevin00} and for all $l\in\mathbb{N}$, there exist constants $\kappa_l,C_l>0$ depending on $l$ and constants $a,b,c>0$ independent of $l$ such that the solution $z_t^z = (q_t,p_t)$ to~\eqref{langevin00} with initial condition $z$ satisfies
\begin{equation}\label{convlaw}
\abs*{ \mathbb{E}[\varphi(z_t^z)] - \tilde{\pi}(\varphi) } \leq C_l e^{-t\kappa_l} \mathcal{K}_l(z) \left\| \frac{\varphi - \tilde{\pi}(\varphi)}{\mathcal{K}_l} \right\|_{L^\infty}
\end{equation}
for Lebesgue almost all initial $z\in\mathbb{R}^{2n}$, $\mathcal{K}_l \geq 1$ given by~\eqref{lyafunc} and all Lebesgue measurable $\varphi$ satisfying
\begin{equation}\label{Linf}
\frac{\varphi}{\mathcal{K}_l} \in L^\infty
\end{equation}
Moreover for any $l\in\mathbb{N}$, $\mathcal{K}_l$ satisfies
\begin{equation}\label{fin}
\int \mathcal{K}_l d\tilde{\pi} < \infty
\end{equation}
and 
\begin{equation}\label{lyapunov}
\mathcal{L}K_l \leq -a_l \mathcal{K}_l + b_l
\end{equation}
for some constants $a_l,b_l >0 $.
\end{theorem}
The proof is from \cite{e19120647}, in which the setting is more general than~\eqref{langevin00} in that the friction matrix is dependent on $q$ and the drift is not necessarily conservative, i.e. the forcing term is not the gradient of a scalar function and the fluctuation-dissipation theorem (see equation (6.2) in \cite{MR3288096}) does not hold, 
but of course, it applies in particular to our setting.
\begin{remark}\label{all}
Inequality~\eqref{convlaw} holds for all initial $z\in\mathbb{R}^{2n}$, as opposed to almost all $z$, given any \textit{bounded} measurable $\varphi$. This is a consequence of combining~\eqref{convlaw} together with the strong Feller property given by Theorem 4.2 in \cite{MR3256873}.
\end{remark}
\begin{proof}
The measure $\tilde{\pi}$ is invariant due to~\eqref{inv}. 
For the rest of the statements, see Theorem 3 in \cite{e19120647}.
\end{proof}
The following corollary holds by taking $\varphi$ as indicator functions and Remark~\ref{all}.
\begin{corollary}\label{tvcor}
Under Assumptions~\ref{smu} and~\ref{assump1}, for all initial $z\in\mathbb{R}^{2n}$, the transition probability $p_t^z$ of~\eqref{langevin00}, given by $p_t^z(A) = \mathbb{P}(z_t^z\in A)$, satisfies
\begin{equation*}
\| p_t^z-\tilde{\pi} \|_{\textrm{TV}} \rightarrow 0 \qquad \textrm{as } t\rightarrow \infty
\end{equation*}
where $\| \cdot \|_{\textrm{TV}}$ denotes the total variation norm.
\end{corollary}
The solution to the Poisson equation is given next following the direction of \cite{MR3069369}.
\begin{theorem}\label{poissonsolve}
Under Assumptions~\ref{smu} and~\ref{assump1}, if $f\in L_0^2(\tilde{\pi})$ satisfies $\frac{f}{\mathcal{K}_l}\in L^\infty$ for some $l\in\mathbb{N}$, then there exists a unique solution $\phi\in L_0^2(\tilde{\pi})$ to the Poisson equation~\eqref{poisson0}. 
Moreover, the solution is given by
\begin{equation}\label{poissol}
\phi = \int_0^\infty P_t(f) dt.
\end{equation}
\end{theorem}
\begin{proof}
For $T>0$, let
\begin{equation*}
g_T := \int_0^T P_t(f) dt.
\end{equation*}
Note that $g_T\in L^2(\tilde{\pi})$ for $T\in\mathbb{R}_+ \cup \{ \infty \}$ and by Theorem~\ref{thmc}
\begin{equation}\label{gTl2}
g_T\rightarrow \int_0^\infty P_t(f) dt
\end{equation}
in $L^2(\tilde{\pi})$ as $T\rightarrow\infty$, specifically~\eqref{convlaw} with $\varphi=f$ and using~\eqref{fin} for $2l$ in place of $l$. Applying $L$, it holds that
\begin{equation*}
Lg_T = \lim_{s\rightarrow 0} \frac{P_s(g_T) - g_T}{s} = \lim_{s\rightarrow 0} \frac{1}{s} \bigg(\!\int_s^{T+s}-\int_0^T\bigg) P_u(f) du = P_T(f) - f,
\end{equation*}
where the exchange in the order of integration is justified by Fubini,~\eqref{convlaw} and the last equality follows by the strong continuity of $(P_t)_{t\geq 0}$ given by Proposition~\ref{strongcont} in Section~\ref{proofs}. Inequalities~\eqref{convlaw} and~\eqref{fin} (with $2l$ in place of $l$) also give
\begin{equation}\label{PTl2}
P_T(f) \rightarrow 0 \qquad \textrm{in } L^2(\tilde{\pi})
\end{equation}
as $T\rightarrow\infty$, so that since $L$ is a closed operator, equations~\eqref{poisson0} and~\eqref{poissol} hold. In addition, $\int\phi d\tilde{\pi} = 0$ follows from the invariance of $\tilde{\pi}$, Theorem~\ref{thmc} and Fubini's theorem.
\end{proof}
We proceed to state the central limit theorem for the solution to~\eqref{langevin00}.
\begin{theorem}\label{theclt}
Under Assumptions~\ref{smu} and~\ref{assump1}, if $f\in L^2(\tilde{\pi})$ satisfies $\frac{f}{\mathcal{K}_l}\in L^\infty$ for some $l\in\mathbb{N}$, the random variable $\frac{1}{\sqrt{t}}\int_0^t (f(z_s) - \pi(f)) ds$ converges in distribution to $\mathcal{N}(0,\sigma_f^2)$ as $t\rightarrow \infty$ for any initial distribution, where
\begin{equation}\label{cltform}
\sigma_f^2 = 2\int \phi (f-\pi(f)) d\tilde{\pi}
\end{equation}
and $\phi\in L_0^2(\tilde{\pi})$ is the solution to~\eqref{poisson0}.
\end{theorem}
\begin{proof}
By Corollary~\ref{tvcor} and Theorem~\ref{poissonsolve}, the result follows by Theorem 2.6 in \cite{MR663900}. See also Theorem 3.1 in \cite{MR3069369}.
\end{proof}

\section{Directional derivative of $\sigma^2$}\label{fdsec}
In this section, we give a number of natural preliminary results that pave the path for the main result in 
Theorem~\ref{funcder}, in which a formula for the derivative~\eqref{funcder00} of $\sigma^2$ with respect to $\Gamma$ is provided. The proofs of Proposition~\ref{distprop}, Lemma~\ref{nabphi} and Theorem~\ref{funcder} are deferred to Section~\ref{proofs}. 

\subsection{Preliminary results and the main formula}\label{premsec}
In order to establish the formula for the directional derivative, heavy use of the differential operator form~\eqref{infgen} for the generator is made. Proposition~\ref{distprop} establishes that $\phi$ solves the Poisson equation also as a partial differential equation, which makes use of the Feynman-Kac representation formula for the solution to the Kolmogorov (backward) equation.

\begin{prop}\label{distprop}
Under Assumptions~\ref{smu} and~\ref{assump1}, if $f\in L_0^2(\tilde{\pi})$ satisfies $\frac{f}{\mathcal{K}_l}\in L^\infty$ for some $l\in\mathbb{N}$, the solution $\phi$ given by~\eqref{poissol} solves $-\mathcal{L}\phi = f$ in the distributional sense for $\mathcal{L}$ given by~\eqref{infgen}, hence classically if in addition $f\in C^\infty$.
\end{prop}
In order for the integral in a formula like~\eqref{funcder00} to be finite, control on the derivatives in $p$ is required. This will also be used in the proof of Theorem~\ref{funcder} and it is given by the following lemma.
\begin{lemma}\label{nabphi}
Under Assumptions~\ref{smu} and~\ref{assump1}, if $f\in L_0^2(\tilde{\pi})$ satisfies $\frac{f}{\mathcal{K}_l}\in L^\infty$ for some $l\in\mathbb{N}$, the weak derivative in $p$ of the solution $\phi$ to $-\mathcal{L}\phi = f$ satisfies
\begin{equation*}
\int \abs*{\nabla\!_p \phi}^2 d\tilde{\pi} < \infty.
\end{equation*}
\end{lemma}
The following preliminary result is about the solution $\phi$ under a momentum reversal. It turns out that this is the solution to the Poisson equation associated to the formal $L^2(\tilde{\pi})$-adjoint $\mathcal{L}^*$ of $\mathcal{L}$ which appears in the proof of Theorem~\ref{funcder}. 
\begin{lemma}\label{tildelem}
Let Assumptions~\ref{smu} and~\ref{assump1} hold and let $f\in C^\infty$ and $\tilde{\phi}\in L_0^2(\tilde{\pi})\cap C^\infty$ be given by~\eqref{phitilde}, 
where $\phi$ is a classical solution to $-\mathcal{L}\phi=f$. Then $\tilde{\phi}$ is a classical solution to the equation
\begin{equation}\label{adjeq}
-\mathcal{L}^* \tilde{\phi} = f,
\end{equation}
where $\mathcal{L}^*$ is the formal $L^2(\tilde{\pi})$-adjoint of $\mathcal{L}$ given by 
\begin{equation*}
\mathcal{L}^* = -p^\top M^{-1} \nabla_q + \nabla U(q)^\top \nabla\!_p - p^\top M^{-1} \Gamma \nabla\!_p + \nabla\!_p^\top \Gamma\nabla\!_p.
\end{equation*}
\end{lemma}
\begin{proof}
The equation $\mathcal{L}^*\tilde{\phi} = \widetilde{\mathcal{L} \phi}$ follows by a straightforward calculation.
\end{proof}
If $f$ is not smooth, the equation~\eqref{adjeq} still holds in the distributional sense, since for $g\in C_c^\infty$ and keeping the notation~\eqref{phitilde} for the momentum reversal on arbitrary functions,
\begin{equation*}
\int f g d\tilde{\pi} = \int f\tilde{g} d\tilde{\pi} = -\int \phi \mathcal{L}^* \tilde{g} d\tilde{\pi} = -\int \phi \widetilde{\mathcal{L}g} d\tilde{\pi} = -\int \tilde{\phi} \mathcal{L} g d\tilde{\pi}.
\end{equation*}

The main formula of this section for the directional derivative of the asymptotic variance is given next. 
The directional derivative of $E:\mathbb{S}_{++}^n\rightarrow\mathbb{R}$ at $\Gamma\in\mathbb{S}_{++}^n$ in the direction $\delta\Gamma\in\mathbb{S}_{++}^n$ is denoted by $
dE(\Gamma).\delta\Gamma = \lim_{\epsilon\rightarrow 0}\frac{1}{\epsilon}(E(\Gamma + \epsilon\delta\Gamma) - E(\Gamma))$ 
whenever the limit exists. 
\begin{theorem}\label{funcder}
Under Assumptions~\ref{smu} and~\ref{assump1}, if $f = f(q) \in L_0^2(\pi)$ is continuous, satisfies $\frac{f}{\mathcal{K}_l}\in L^\infty$ for some $l\in\mathbb{N}$ and there exists $\epsilon'>0$ such that $\Gamma,\Gamma + \epsilon\delta\Gamma\in \mathbb{S}_{++}^n$ for $0<\epsilon \leq \epsilon'$, then the directional derivative of the asymptotic variance $\sigma^2$ at $\Gamma$ in the direction $\delta\Gamma$ has the form
\begin{equation}\label{funcdereq}
d\sigma^2.\delta\Gamma = -2 \int (\nabla\!_p\phi)^\top \delta\Gamma \nabla\!_p \tilde{\phi} d\tilde{\pi},
\end{equation}
where $\phi$ is the solution~\eqref{poissol} to the Poisson equation for the dynamics~\eqref{langevin00} at $\Gamma$ and $\tilde{\phi}$ is given by~\eqref{phitilde}.
\end{theorem}
As mentioned in the introduction, from~\eqref{funcdereq}, the direction~\eqref{DeltaGam} 
guarantees a decrease in asymptotic variance; similarly the scalar change in $\Gamma$ given by~\eqref{DeltaGam} where the outer product is replaced by a dot product guarantees a decrease in~$\sigma^2$.

\subsection{A formula using a tangent process}\label{tangentproc}
Equation~\eqref{funcdereq} has a more useful form. The first variation process of~\eqref{langevin00} is used here to calculate~\eqref{nabphi0}; this will be the main methodology used in the numerical sections. This alternative formula given in Theorem~\ref{newformprop} provides a way to avoid using a finite difference Monte Carlo estimate of the derivative of an expectation. For simplicity, we set $M = I_n$ here. The first variation process associated to~\eqref{langevin00}, denoted by $(D_pq_t,D_pp_t)\in\mathbb{R}^{n\times 2n}$ for $t\geq 0$, is defined as the matrix-valued solution to
\begin{equation} \label{Dq1}
\partial_t\begin{pmatrix} D_p q_t \\
D_p p_t \end{pmatrix} = \begin{pmatrix} 0 & I_n \\
-D^2 U(q_t) & -\Gamma \end{pmatrix}
\begin{pmatrix}
D_p q_t \\
D_p p_t
\end{pmatrix}
\end{equation}
with the initial condition $D_p q_0 = 0$, $D_p p_0 = I_n$. 
By Theorem 39 of Chapter V in \cite{MR2020294}, the partial derivatives of $(q_t,p_t)$ with respect to the initial values in $p$ indeed satisfy~\eqref{Dq1} and $(D_p q_t,D_p p_t)$ is continuous with respect to those initial values. Note that there exists a unique solution to~\eqref{Dq1} by Theorem 38 in the same chapter of \cite{MR2020294}. We omit the notational dependence of $(q_t,p_t)$ on its initial condition $(q_0,p_0) = (q,p) = z$ whenever convenient in the following.
\begin{theorem}\label{newformprop}
Let Assumptions~\ref{smu} and~\ref{assump1} hold. If in addition,
\begin{itemize}
\item there exist $U_0 >0$ and $Q\in \mathbb{S}_{++}^n$ such that for all $q\in\mathbb{R}^n$, $v\in\mathbb{R}^n$,
\begin{equation*}
v^\top D^2 U(q)v \geq U_0 \abs*{v}^2,\quad
D^2 U(q) = Q + D(q),
\end{equation*}
where $D:\mathbb{R}^n\rightarrow\mathbb{R}^{n\times n}$ is small enough everywhere, in particular,
\begin{equation}\label{bdquad}
\abs*{D(q)} \leq \hat{\lambda} := \min\bigg(\frac{\lambda_m}{2}, \frac{\lambda_m U_0^2}{8\lambda_M^2}, \frac{\lambda_m U_0}{16},\frac{U_0}{8}\sqrt{\sigma_{\min}(Q)}\bigg),
\end{equation}
where $\lambda_m,\lambda_M>0$ are respectively the smallest and largest eigenvalue of $\Gamma$ and $\sigma_{\min}(Q)$ denotes the smallest eigenvalue of $Q$, is assumed;
\item $f = f(q) \in L_0^2(\pi)$ is everywhere differentiable and 
satisfies $\frac{\abs*{f}+\abs*{\nabla f}}{\mathcal{K}_l}\in L^\infty$ for some $l\in\mathbb{N}$,
\end{itemize}
then the weak derivative $\nabla\!_p\phi$ has the form
\begin{equation}\label{newform}
\nabla\!_p \phi(q,p) = \int_0^\infty \mathbb{E}[\nabla f(q_t)^\top D_p q_t] dt,
\end{equation}
where $q_t$ solves~\eqref{langevin00} with initial condition $(q_0,p_0) = (q,p)$ and $D_p q_t$ solves~\eqref{Dq1}, the latter satisfying
\begin{equation}\label{dqdec}
\abs*{D_p q_t}^2+\abs*{D_p p_t}^2 \leq C'e^{-Ct}
\end{equation}
for some constants $C,C'>0$ independent of $(q_0,p_0)$ and $\omega\in\Omega$.
\end{theorem}
The assumptions on $U$ are made in order to ensure that the process $(D_p q_t,D_p p_t)$ converges to zero exponentially quickly in order for the integral in~\eqref{newform} to be finite. Specifically, $U$ is assumed to be close to some particular quadratic function $q^\top Q q$, cf. \cite{MR2731396}. 
\begin{remark}
Exponential decay of the first variation process is not required, only some uniform (in initial $(q_0,p_0)$) integrability in time of $D_pq_t$ together with a boundedness assumption on $\nabla f$. On the other hand, Proposition 1 in \cite{MR4091098} and Proposition 4 in \cite{monmarche2021sure} explores more detailed conditions under which contractivity holds and does not hold.
\end{remark}
\begin{proof}
Let $b>0$ be the constant
\begin{equation*}
b = \min\bigg( \frac{\lambda_m U_0}{2\lambda_M^2}, \frac{\lambda_m}{4},\frac{1}{2}\sqrt{\sigma_{\min}(Q)}\bigg)
\end{equation*}
so that $\hat{\lambda}$ reduces to 
\begin{equation*}
\hat{\lambda} = \min\bigg(\frac{\lambda_m}{2},b\frac{U_0}{4}\bigg)
\end{equation*}
and we have the following bound
\begin{align}
&\frac{1}{2}\partial_t \bigg[e_i^\top\begin{pmatrix} D_p q_t \\
D_p p_t \end{pmatrix}^\top \begin{pmatrix} Q & b I_n \\
b I_n & I_n \end{pmatrix}\begin{pmatrix} D_p q_t \\
D_p p_t \end{pmatrix}e_i\bigg]\nonumber\\
&= e_i^\top D_p q_t^\top Q D_p p_t e_i + b\abs*{D_p p_t e_i}^2\nonumber\\
&\quad - e_i^\top(bD_p q_t + D_p p_t)^\top (D^2U(q)D_p q_t + \Gamma D_p p_t)e_i\nonumber\\
&= -be_i^\top D_p q_t^\top D^2 U(q_t) D_p q_t e_i + e_i^\top D_p q_t^\top ( - b\Gamma - D(q_t) ) D_p p_t e_i\nonumber\\
&\quad - e_i^\top D_p p_t^\top (\Gamma - b I_n) D_p p_t e_i\nonumber\\
&\leq \bigg(- b U_0 + \frac{b U_0}{2} + \frac{\hat{\lambda}}{2} \bigg) \abs*{D_p q_t e_i}^2 + \bigg( -\lambda_m + b + \frac{b\lambda_M^2}{2U_0} + \frac{\hat{\lambda}}{2} \bigg)\abs*{D_p p_t e_i}^2\nonumber\\
&\leq -\frac{bU_0}{4}\abs{D_p q_t e_i}^2 - \frac{\lambda_m}{4} \abs*{D_p p_t e_i}^2\nonumber\\
&\leq -C e_i^\top \begin{pmatrix} D_p q_t \\
D_p p_t \end{pmatrix}^\top \begin{pmatrix} Q & b I_n \\
b I_n & I_n \end{pmatrix}\begin{pmatrix} D_p q_t \\
D_p p_t \end{pmatrix} e_i\label{dqdec}
\end{align}
for some generic constant $C>0$ independent of the initial values $(q_0,p_0)$ and $\omega\in\Omega$. Consequently, using the (weighted) boundedness assumption on $\abs*{\nabla f}$ and for each index $i$,
\begin{align}
\abs*{(\nabla f(q_t)^\top D_p q_t)_i } &\leq C'e^{-Ct} \abs*{\nabla f(q_t)}\nonumber\\
&\leq  C'e^{-Ct}(\abs*{\nabla f(q_t)} - \pi(\abs*{\nabla f}))+ C'e^{-Ct}\label{nabfdq}
\end{align}
for a generic $C'>0$ independent of $(q_0,p_0)$ and $\omega\in\Omega$. 
Due to~\eqref{nabfdq} together with Fubini's theorem, it holds for $T>0$ and a test function $g\in C_c^\infty(\mathbb{R}^{2n})$ that
\begin{align*}
\int\int_0^T\mathbb{E}[f(q_t^z)]dt \nabla g(z) dz &= \int_0^T\mathbb{E}[\int f(q_t^z)\nabla g(z)dz ]dt\\
&= -\int_0^T\mathbb{E}[\int\nabla f(q_t^z)^\top D_p q_t^z g(z) dz]dt \\
&= -\int\int_0^T \mathbb{E}[\nabla f(q_t^z)^\top D_p q_t^z] dt  g(z) dz.
\end{align*}
Using Theorem~\ref{thmc},~\eqref{nabfdq} again and dominated convergence to take $T\rightarrow\infty$ on both sides concludes the proof.

\end{proof}

At any $t$, equation~\eqref{newform} can be used in a practical way in order to estimate the gradient direction~\eqref{DeltaGam}. Specifically, the following estimator can be used. Given $(q,p)\in\mathbb{R}^{2n}$, which we think of as solutions $(q_t,p_t)$ to~\eqref{langevin00} hence approximately distributed according to $\tilde{\pi}$, let
\begin{equation}\label{expresdel}
\delta\Gamma = \int_0^T\nabla f(q_t^{(q,p)})^\top D_pq_t^{(q,p)} dt \otimes \int_0^T\nabla f(q_t^{(q,-p)})^\top D_pq_t^{(q,-p)} dt
\end{equation}
for any large enough $T>0$, where $(q_t^{(q,p)},p_t^{(q,p)})$, $(q_t^{(q,-p)},p_t^{(q,-p)})$ solve~\eqref{langevin00} with independent realisations of $W_t$ and $D_p$ has been used to denote the derivative with respect to the initial $p$ as above. 
The next result is that this estimator has finite variance. Note that for $(q,p)$ distributed away from stationarity, the estimator cannot be expected to be an unbiased estimator of~\eqref{DeltaGam}.

\begin{theorem}
Let the assumptions of Theorem~\ref{newformprop} hold. For Lebesgue almost-all $(q,p)\in\mathbb{R}^{2n}$, each entry of $\delta\Gamma$ defined in~\eqref{expresdel} has finite variance.
\end{theorem}
\begin{proof}
It suffices to show that~\eqref{expresdel} has finite second moment, for which it suffices to show that each element in the vector of time integrals
\begin{equation*}
\int_0^T\nabla f(q_t^{(q,p)})^\top D_p q_t^{(q,p)} ds
\end{equation*}
has finite second moments by independence. For each index $i$, using~\eqref{dqdec},
\begin{align}
\abs*{(\nabla f(q_t)^\top D_p q_t)_i }^2 &\leq C'^2e^{-2Ct} \abs*{\nabla f(q_t)}^2\nonumber\\
&\leq  C'^2e^{-2Ct}(\abs*{\nabla f(q_t)}^2 - \pi(\abs*{\nabla f}^2))+ C'^2e^{-2Ct}\pi(\abs*{\nabla f}^2),\nonumber
\end{align}
so that using the (weighted) boundedness assumption on $\abs*{\nabla f}$ together with Theorem~\ref{thmc} and Fubini's theorem, the proof concludes.
\end{proof}

\section{Quadratic cases}\label{anacas}
Throughout this section, the target measure $\pi$ is assumed to be Gaussian, when mean zero this is given by $\pi\propto \exp(-\frac{1}{2}q^\top\Sigma^{-1}q)$ for $\Sigma\in\mathbb{S}_{++}^n$, in other words, the potential is quadratic, $U(q) = \frac{1}{2}q^\top\Sigma^{-1}q$. For polynomial observables, we look for solutions to the Poisson equation by using a polynomial ansatz and comparing coefficients in order to obtain an explicit expression for the asymptotic variance. The results provide benchmarks to test the performance of the algorithms that arise from using the gradient in Theorem~\ref{funcder} as well as intuition for how $\Gamma$ can be improved in concrete cases. We will consider the following cases.
\begin{enumerate}
\item Quadratic $f = \frac{1}{2}q^\top U_0 q$, given commutativity between $U_0$ and $\Sigma$ (Proposition~\ref{comprop}), also $f = \frac{1}{2} U_0 q^2 + l q$ in one dimension (Proposition~\ref{quadprop1});
\item Odd polynomial $f$, where the asymptotic variance will be shown to decrease to zero as $\Gamma \rightarrow 0$ (Proposition~\ref{linprop}, Corollary~\ref{lincor} and Proposition~\ref{oddprop});
\item Quartic $f$ in one dimension, in which the situation is similar to quadratic $f$ (Proposition~\ref{quarprop});
\end{enumerate}
We proceed with stating in more detail the general situation of this section. Let $\Sigma\in\mathbb{S}_{++}^n$, $U_0\in\mathbb{S}_{++}^n$ and $l\in\mathbb{R}^n$. The Gaussian invariant measure $\tilde{\pi}$ and the observable $f:\mathbb{R}^{2n}\rightarrow\mathbb{R}$ are given by
\begin{equation}\label{quadob}
\tilde{\pi}\propto 
\exp\bigg(\!-\frac{1}{2}q\cdot \Sigma^{-1}q - \frac{1}{2}p\cdot M^{-1} p\bigg),\quad f(q) = \frac{1}{2}q\cdot U_0q + l\cdot q
\end{equation}
and 
the value $\pi(f)$ becomes
\begin{equation}\label{goal}
\pi(f) = \int f d\pi = \int \frac{1}{2}q\cdot U_0q d\pi = \frac{1}{2} U_0:\Sigma.
\end{equation}
The infinitesimal generator $\mathcal{L}$ becomes in this case
\begin{align}
\mathcal{L} &= 
\begin{pmatrix}
0 & M^{-1}\\
-\Sigma^{-1} & -\Gamma M^{-1}
\end{pmatrix}
\begin{pmatrix}
q\\
p
\end{pmatrix}\cdot
\nabla + \nabla_{\! p}\cdot\Gamma\nabla_{\! p}  \nonumber \\
&= M^{-1}p \cdot \nabla_q - \Sigma^{-1}q \cdot \nabla_p - \Gamma  M^{-1}p\cdot \nabla_p + \nabla_p\cdot\Gamma\nabla_p. \label{generator}
\end{align}
Consider the natural candidate solution $\phi$ to the Poisson equation~\eqref{poisson0} given by
\begin{equation}\label{ansatz}
\phi(q,p) = \frac{1}{2}q\cdot Gq + q\cdot Ep + \frac{1}{2}p\cdot Hp + g\cdot q + h\cdot p - \frac{1}{2}(G:\Sigma + H:M).
\end{equation}
for some constant matrices $G,E,H\in\mathbb{R}^{n\times n}$ and vectors $g,h\in\mathbb{R}^n$.
\begin{lemma}
Given $f$ in~\eqref{quadob}, $\pi(f)$ in~\eqref{goal} and $\mathcal{L}$ of the form~\eqref{generator}, $\phi$ given by~\eqref{ansatz} is a solution to the Poisson equation~\eqref{poisson0} if and only if
\begin{align}
\Sigma^{-1}q\cdot(E^\top q + h) - \Gamma:H_S - \frac{1}{2}q\cdot U_0 q - l\cdot q + \frac{1}{2}U_0 :\Sigma &= 0,\label{zery}\\
- M^{-1} (G_Sq + g) + H_S\Sigma^{-1}q + M^{-1}\Gamma(E^\top q + h) &= 0,\label{firy}\\
-E^\top M^{-1} + H_S\Gamma M^{-1} &= A_1,\label{secy}
\end{align}
for some antisymmetric $A_1\in\mathbb{R}^{n\times n}$.
\end{lemma}

\begin{proof}
Substituting~\eqref{generator},~\eqref{ansatz} and~\eqref{quadob} into the Poisson equation~\eqref{poisson0}, one obtains
\begin{align*}
&-\begin{pmatrix}
0 & M^{-1}\\
-\Sigma^{-1} & -\Gamma M^{-1}
\end{pmatrix}
\begin{pmatrix}
q\\
p
\end{pmatrix}\cdot
\begin{pmatrix}
G_Sq + Ep + g\\
E^\top q + H_S p + h
\end{pmatrix} -\Gamma:H_S\\
&\qquad= \frac{1}{2}q\cdot U_0 q + l\cdot q -\frac{1}{2}U_0 :\Sigma.
\end{align*}
Comparing the constant, first order and second order coefficients in $p$ give respectively the sufficient conditions~\eqref{zery},~\eqref{firy} and~\eqref{secy} as stated.
\end{proof}

\subsection{Quadratic observable}
Similar calculations in this situation have appeared previously in Proposition 1 in \cite{MR3483241}, where explicit expressions analogous to $G$, $E$, $H$ and for $\sigma^2$ are given. 
For our purposes of finding an optimal $\Gamma$, instead of taking these explicit expressions, we keep unknown antisymmetric matrices (such as $A_1$) as they appear and eventually use commutativity between $\Sigma$ and $U_0$ to show that the antisymmetric matrices are zero. We continue from~\eqref{zery},~\eqref{firy} and~\eqref{secy} with finding explicit expressions for the coefficients $G$, $E$, $H$ of $\phi$. 
\begin{lemma}\label{licit}
Given $f$ in~\eqref{quadob}, $\pi(f)$ in~\eqref{goal}, $\mathcal{L}$ of the form~\eqref{generator}
, $\phi$ given by~\eqref{ansatz} is a solution to the Poisson equation~\eqref{poisson0} with~\eqref{generator} if and only if
\begin{align}
G_S &= \frac{1}{2}M(\Sigma U_0  - \Sigma A_2-2A_1M)\Gamma^{-1}\Sigma^{-1} + \frac{1}{2} \Gamma(U_0  \Sigma - A_2\Sigma ),\label{gexp}\\
E &= \frac{1}{2}U_0 \Sigma + \frac{1}{2}A_2\Sigma,\label{eexp}\\
H_S &= \frac{1}{2}(\Sigma U_0  - \Sigma A_2 - 2A_1M)\Gamma^{-1},\label{hexp}\\
h &= \Sigma l \qquad\textrm{and}\qquad
g = \Gamma\Sigma l.\label{pqexp}
\end{align}
for some antisymmetric matrices $A_1,A_2$.
\end{lemma}
\begin{proof}
Comparing coefficients in $q$ in equation~\eqref{zery} gives
\begin{align}
2\Gamma : H_S &= U_0 :\Sigma\label{6}\\
h^\top \Sigma^{-1} &= l^\top \label{7}\\
2E\Sigma^{-1} &= U_0  + A_2\label{8}
\end{align}
and the same for condition~\eqref{firy} gives
\begin{align}
M^{-1}G_S &= H_S\Sigma^{-1} + M^{-1}\Gamma E^\top,\label{9}\\
M^{-1}g &= M^{-1}\Gamma h.\label{10}
\end{align}
Condition~\eqref{8} yields~\eqref{eexp}. Together with~\eqref{secy}, this gives~\eqref{hexp}. 
From the expression~\eqref{hexp} and by symmetry of $U_0 $, condition~\eqref{6} is in turn satisfied:
\begin{align*}
2\Gamma : H_S &= \Gamma : ((\Sigma U_0  - \Sigma A_2 - 2A_1M)\Gamma^{-1})\\
&= \sum_{i,j,k,l} \Gamma_{ji}(\Sigma_{ik} (U_0)_{kl} - \Sigma_{ik}(A_2)_{kl} - (A_1)_{ik}M_{kl}) (\Gamma^{-1})_{lj}\\
&=  \sum_{i,k} (U_0)_{ki} \Sigma_{ki} = U_0 :\Sigma,
\end{align*}
where symmetry of $\Sigma$ and $M$ have been used. Substituting~\eqref{eexp} and~\eqref{hexp} into equation~\eqref{9} then gives~\eqref{gexp}. Equations~\eqref{7} and~\eqref{10} give the equations~\eqref{pqexp} for $g$ and $h$.
\end{proof}

The asymptotic variance from Theorem~\ref{theclt} can be given by a formula in terms of $\Sigma,U_0$ and the coefficients of $\phi$. Before substituting the expressions from Lemma~\ref{licit} into the formula, we give the formula itself, which is adapted from the proof of Proposition 1 in \cite{MR3483241}.
\begin{lemma}\label{av}
If the solution $\phi$ to the Poisson equation~\eqref{poisson0} for $f$ given by~\eqref{quadob}, $\pi(f)$ given by~\eqref{goal} and $\mathcal{L}$ given by~\eqref{generator} is of the form~\eqref{ansatz}, the asymptotic variance $\sigma^2$ given by~\eqref{cltform} has the expression
\begin{equation}\label{asympfor}
2\langle \phi, f-\pi(f) \rangle_{\tilde{\pi}} = \textrm{Tr}(G_S\Sigma U_0  \Sigma)+ 2g\cdot \Sigma l.
\end{equation}
\end{lemma}

\begin{proof}
Denote
\begin{align*}
\bar{G} = \begin{pmatrix}
G_S & E\\
E^\top & H_S
\end{pmatrix},\quad
\bar{U}_0 = \begin{pmatrix}
U_0  & 0\\
0 & 0
\end{pmatrix}, \quad
\bar{\Sigma} = \begin{pmatrix}
\Sigma & 0\\
0 & M
\end{pmatrix},\quad 
\bar{g} = \begin{pmatrix}
g\\
h
\end{pmatrix},\quad
\bar{l} = \begin{pmatrix}
l\\
0
\end{pmatrix}.
\end{align*}
Each of $\phi$ and $f-\pi(f)$ are given by
\begin{align*}
\phi(z) &= \frac{1}{2}z\cdot \bar{G}z - \bar{g}\cdot z-\frac{1}{2}\bar{G}:\bar{\Sigma}\\
f(z)-\pi(f) &= \frac{1}{2}z\cdot \bar{U}_0 z - \bar{l}\cdot z-\frac{1}{2}\bar{U}_0:\bar{\Sigma}
\end{align*}
for $z=(q,p)\in\mathbb{R}^{2n}$. Substituting into $\sigma^2=2 \langle \phi, f-\pi(f) \rangle_{\tilde{\pi}}$ gives 
\begin{align*}
&2 \int \phi (f-\pi(f))d\tilde{\pi}\\
&\qquad= \frac{1}{2}\int (z\cdot \bar{G}z) (z\cdot \bar{U}_0 z) d\tilde{\pi} - \frac{1}{2}\int( z\cdot \bar{G}z) \bar{U}_0:\bar{\Sigma} d\tilde{\pi} + 2\int (\bar{g}\cdot z )(\bar{l}\cdot z) d\tilde{\pi}\\
&\qquad\quad - \frac{1}{2}\int \bar{G}:\bar{\Sigma} ( z\cdot \bar{U}_0 z)d\tilde{\pi} + \frac{1}{2}(\bar{G}:\bar{\Sigma})(\bar{U_0 }:\bar{\Sigma}),
\end{align*}
where
\begin{align*}
\int (z\cdot \bar{G}z) (z\cdot \bar{U}_0 z) d\tilde{\pi}
&= \sum_{i,j,u,v} \bar{G}_{ij}(\bar{U}_0)_{uv}\int z_i z_j z_u z_v d\tilde{\pi}\\
&= \sum_{i,j,u,v} \bar{G}_{ij}(\bar{U}_0)_{uv} \Big( \bar{\Sigma}_{ij}\bar{\Sigma}_{uv} + \bar{\Sigma}_{iu}\bar{\Sigma}_{jv} + \bar{\Sigma}_{iv}\bar{\Sigma}_{ju} \Big)\\
&= (\bar{G}:\bar{\Sigma})(\bar{U}_0:\bar{\Sigma}) + 2\textrm{Tr}(\bar{G}\bar{\Sigma}\bar{U}_0\bar{\Sigma}).
\end{align*}
As a result,
\begin{align*}
2\int \phi(f-\pi(f)) d\tilde{\pi} &= \frac{1}{2}(\bar{G}:\bar{\Sigma})(\bar{U}_0:\bar{\Sigma}) + \textrm{Tr}(\bar{G}\bar{\Sigma}\bar{U}_0\bar{\Sigma}) - \frac{1}{2}(\bar{G}:\bar{\Sigma})(\bar{U}_0:\bar{\Sigma}) \\
&\quad + 2\int( \bar{g}\cdot z)( \bar{l}\cdot z)d\tilde{\pi}\\[0.5em]
&= \textrm{Tr}(\bar{G}\bar{\Sigma}\bar{U}_0\bar{\Sigma}) + 2 \bar{g}\cdot \bar{\Sigma} \bar{l}.
\end{align*}

\end{proof}

From the expressions~\eqref{gexp} and~\eqref{hexp} for $G_S$ and $H_S$ respectively, it is not straightforward to check that there exist antisymmetric $A_1$ and $A_2$ such that the right hand sides are indeed symmetric at this point, which is necessary for the ansatz~\eqref{ansatz} for $\phi$ to be a valid solution. On the other hand, if $\Sigma$, $U_0$, $\Gamma$, $M$ all commute, then the right hand sides of~\eqref{gexp} and~\eqref{hexp} are symmetric for $A_1=A_2=0$ and the coefficients $G$ and $H$ become explicit, which allows taking derivatives of $\sigma^2$ with respect to the entries of $\Gamma$. Moreover, the explicit coefficients allow optimisation of $M$, which is given by the following proposition. 
\begin{prop}\label{mprop}
Let $\Sigma U_0 = U_0 \Sigma$, $f$ be as in~\eqref{quadob}, $\pi(f)$ be as in~\eqref{goal}, $\mathcal{L}$ be of the form~\eqref{generator} 
and $\phi$ be the solution to the Poisson equation~\eqref{poisson0}. The following holds.
\begin{equation*}
\lim_{M\rightarrow 0}\int \phi (f-\pi(f))d\tilde{\pi} = \inf_{M\in\mathbb{S}_{++}^n}\int \phi (f-\pi(f))d\tilde{\pi},
\end{equation*}
where the limit on the left hand side is in the sense of $M = mI_n$, $m\rightarrow 0^+$.
\end{prop}
\begin{proof}
Let 
\begin{subequations}\label{phicos}
\begin{align}
G &= \frac{1}{2}M\Sigma U_0\Gamma^{-1}\Sigma^{-1} + \frac{1}{2}\Gamma U_0 \Sigma, \quad E = \frac{1}{2}U_0 \Sigma,\quad H= \frac{1}{2}\Sigma U_0 \Gamma^{-1},\\
g &= \Gamma \Sigma l,\quad h = \Sigma l
\end{align}
\end{subequations}
so that by Lemma~\ref{licit}, $\phi$ given by~\eqref{ansatz} is the solution to the Poisson equation~\eqref{poisson0} and inserting $G,g$ into~\eqref{asympfor} gives
\begin{equation}\label{subasym}
2\langle \phi ,f -\pi(f) \rangle_{\tilde{\pi}} = \frac{1}{2}\textrm{Tr}(M\Sigma U_0 \Gamma^{-1}U_0  \Sigma + \Gamma U_0  \Sigma^2 U_0  \Sigma) + 2l^\top \Sigma \Gamma \Sigma l. 
\end{equation}
The result follows since $A:B >0$ for $A,B\in\mathbb{S}_{++}^n$.
\end{proof}

Proposition~\ref{mprop} solves the optimisation problem in $M$ in the stated setting, but highlights the corresponding discrete time problem since one cannot take $M^{-1}\rightarrow\infty$ in practice. We focus on the optimisation of $\Gamma$ and fix $M = I_n$ in the following.

\begin{prop}\label{comprop}
Let $\Sigma, U_0, l, M$ be such that
\begin{equation}\label{comcond}
\Sigma U_0 = U_0\Sigma ,\qquad l = 0,\qquad M = I_n,
\end{equation}
$f$ be as in~\eqref{quadob}, $\pi(f)$ be as in~\eqref{goal}, $\mathcal{L}$ be of the form~\eqref{generator} and $\phi$ be the solution to the Poisson equation~\eqref{poisson0}. The following holds.
\begin{equation*}
\min_{\Gamma\in\mathbb{S}_{\Sigma}} 2\int \phi (f-\pi(f)) d\tilde{\pi} = \textrm{Tr}\Big(U_0^2\Sigma^{\frac{5}{2}}\Big),
\end{equation*}
where $\mathbb{S}_{\Sigma}$ is the set of symmetric positive definite matrices commuting with $\Sigma$ and the minimum is attained by 
\begin{equation*}
\Gamma = \Sigma^{-\frac{1}{2}}.
\end{equation*}
\end{prop}
\begin{proof}
Let $\Sigma = P^\top \Sigma_d P$  be the eigendecomposition of $\Sigma$ for orthogonal $P$. 
Since all symmetric matrices in the set commuting with $\Sigma$ share eigenvectors with $\Sigma$, it suffices to find a unique extremal point of the asymptotic variance with respect to the eigenvalues of $\Gamma$, call them $(\lambda_i)_{1\leq i \leq n}$, $\lambda_i\geq 0$. 
Setting again~\eqref{phicos}, $\phi$ given by~\eqref{ansatz} is the solution to the Poisson equation~\eqref{poisson0} and the asymptotic variance $\sigma^2$ given by~\eqref{cltform} becomes 
\begin{equation}\label{subasym1}
2\langle \phi ,f -\pi(f) \rangle_{\tilde{\pi}} = \frac{1}{2}\textrm{Tr}(\Sigma U_0 \Gamma^{-1}U_0  \Sigma + \Gamma U_0  \Sigma^2 U_0  \Sigma),
\end{equation}
which reduces to a sum of functions of the form $a_i \lambda_i^{-1} + b_i \lambda_i$, $a_i,b_i>0$ after diagonalising with $P$ and the result follows.
\end{proof}

In the scalar case, we can remove the restriction on $l$.
\begin{prop}\label{quadprop1}
If 
$n=1$, $U_0 \neq 0$, $l \neq 0$, $f:\mathbb{R}\rightarrow\mathbb{R}$ is given by~\eqref{quadob}, $\pi(f)$ is given by~\eqref{goal}, $\mathcal{L}$ is of the form~\eqref{generator} and $\phi$ is the solution to the Poisson equation~\eqref{poisson0}, then 
\begin{equation*}
\min_{\Gamma >0} 2\int \phi (f-\pi(f)) d\tilde{\pi} =  M^{\frac{1}{2}}\Sigma^2U_0^2(\Sigma + 4l^2 U_0^{-2})^{\frac{1}{2}}
\end{equation*}
and the minimum is attained by
\begin{equation}\label{og}
\Gamma = \frac{M^{\frac{1}{2}}}{(\Sigma + 4l^2 U_0^{-2})^{\frac{1}{2}}}.
\end{equation}
\end{prop}
\begin{proof}
By Lemma~\ref{licit}, the solution~\eqref{ansatz} to the Poisson equation~\eqref{poisson0} is
\begin{equation*}
\phi = \bigg( \frac{U_0 \Gamma\Sigma}{4} + \frac{M U_0 }{4\Gamma} \bigg) q^2 + \frac{U_0 \Sigma}{2} qp + \frac{U_0 \Sigma}{4\Gamma} p^2 + \Sigma \Gamma l q + \Sigma l p - \frac{U_0 \Gamma\Sigma^2}{4} - \frac{ M U_0 \Sigma}{2\Gamma}.
\end{equation*}
By Lemma~\ref{av}, the asymptotic variance is given by
\begin{equation*}
2\int \phi (f-\pi(f)) d\tilde{\pi} = 2\Sigma^2 \bigg(\frac{U_0^2\Sigma}{4} + l^2\bigg)\Gamma + \frac{U_0^2\Sigma^2}{2\Gamma},
\end{equation*}
which attains the stated minimum at~\eqref{og} as claimed.
\end{proof}

\subsection{Odd polynomial observable}\label{oddpol}

Another special case within~\eqref{quadob} where the solution $\phi$ can be readily identified is when 
\begin{equation*}
U_0 =0,
\end{equation*}
that is, for linear observables.
More generally, (almost) zero variance can be attained in the following special case.

\begin{prop}\label{linprop}
Under Assumption~\ref{smu} and~\ref{assump1}, for a general target measure $\pi \propto e^{-U}$ on $\mathbb{R}^n$,
if the observable $f$ 
is of the form
\begin{equation*}
f(q) = \alpha \cdot \nabla U,
\end{equation*}
for $\alpha = (\alpha_1,\dots,\alpha_n)$, $\alpha_i\in\mathbb{R}$, $\mathcal{L}$ is of the general form~\eqref{infgen} and $\phi$ is the solution to the Poisson equation~\eqref{poisson0}, then the asymptotic variance satisfies
\begin{equation}\label{linpropdisp}
\inf_{\Gamma\in \{\gamma I_n : \gamma>0\}} 2\int \phi (f-\pi(f))d\tilde{\pi} = 0.
\end{equation}
\end{prop}

\begin{proof}
Let $\Gamma = \gamma I_n$, $\gamma\in\mathbb{R}$.
Note there is a unique solution $\phi\in L_0^2(\tilde{\pi})$ to~\eqref{poisson0} by Theorem~\ref{poissonsolve}. The solution $\phi$ to~\eqref{poisson0} has the expression
\begin{equation*}
\phi = \sum_i \alpha_i (\gamma q_i + p_i ).
\end{equation*}
The asymptotic variance is equal to
\begin{align*}
2\langle \phi,f - \pi(f)\rangle_{\tilde{\pi}} &= 2\gamma\sum_{i,j}\alpha_i \alpha_j \int_{\mathbb{R}^n}q_i \partial_{q_j} U(q) \pi(dq)\\
&= -2\gamma\sum_{i}\alpha_i^2 \int_{\mathbb{R}^n}q_i \partial_{q_i} \pi(q) dq - 2\gamma\sum_{i\neq j}\alpha_i \alpha_j \int_{\mathbb{R}^n}q_i \partial_{q_j} \pi(q) dq\\
&= 2\gamma\sum_{i}\alpha_i^2\int_{\mathbb{R}^n} \pi(q) dq - 2\gamma\sum_{i\neq j}\alpha_i \alpha_j \int_{\mathbb{R}^{n-1}}q_i \int_{\mathbb{R}}\partial_{q_j} \pi(q) dq_j dq_{-j} \\
&= 2\gamma\sum_{i}\alpha_i^2.
\end{align*}
where $dq_{-j}$ denotes $dq_1\dots dq_{j-1}dq_{j+1}\dots dq_n$.
\end{proof}

\begin{corollary}\label{lincor}
Given a Gaussian target measure with density $\pi\propto e^{-U}$ on $\mathbb{R}^n$, observable $f:\mathbb{R}^n\rightarrow\mathbb{R}$ as in~\eqref{quadob} with $U_0=0$, that is,
\begin{equation*}
f(q) = l \cdot q,
\end{equation*}
where $l\in\mathbb{R}^n$, $\pi(f)=0$, $\mathcal{L}$ of the form~\eqref{infgen} and $\phi$ the solution to the Poisson equation~\eqref{poisson0}, equation~\eqref{linpropdisp} holds.
\end{corollary}

There is some intuition in the situation in Corollary~\ref{lincor}. First note that the Langevin diffusion with $\Gamma = 0$ reduces to deterministic Hamiltonian dynamics and that it is the limit case for the $\Gamma$ attaining arbitrarily small asymptotic variance in the proof of Proposition~\ref{linprop}. The result indicates that this is optimal in the linear observable, Gaussian measure case (i.e.~\eqref{quadob}, $U_0 = 0$) and this aligns with the fact that the value~\eqref{goal} to be approximated is exactly the value at the $q=p=0$, so that Hamiltonian dynamics starting at $q=0$, staying there for all time, approximates the integral~\eqref{goal} with perfect accuracy. A similar idea holds for when the starting point is not $q=p=0$, where~\eqref{goal} is approximated exactly after any integer number of orbits in $(q,p)$ space. Continuing on this idea, it seems reasonable that the same statement holds more generally for any odd observable. At least, the following holds in one dimension.

\begin{prop}\label{oddprop}
If $n=1$, $\hat{k}\in\mathbb{N}$ and $f:\mathbb{R}\rightarrow\mathbb{R}$ is an odd finite order polynomial observable given by
\begin{equation}\label{genodd}
f(q) = \sum_{i=0}^{\hat{k}} a_i q^{2i+1},
\end{equation}
$\pi(f) = 0$, $\mathcal{L}$ is of the form~\eqref{generator} and $\phi$ is the solution to the Poisson equation~\eqref{poisson0}, then the asymptotic variance satisfies~\eqref{linpropdisp}.
\end{prop}

The proof of Proposition~\ref{oddprop} is deferred to Section~\ref{proofs}.

\subsection{Quartic observable}
The situation in the quartic observable case, at least in one dimension, is similar to quadratic observable case.
\begin{prop}\label{quarprop}
If 
$n=1$ and $f:\mathbb{R}\rightarrow\mathbb{R}$ is a quartic observable given by
\begin{equation}\label{quar}
f(q) = q^4,
\end{equation}
$\pi(f) = 3\Sigma^2$ for some $\Sigma>0$, $\mathcal{L}$ is of the form~\eqref{generator}, $M=1$ and $\phi$ is the solution to the Poisson equation~\eqref{poisson0}, then there exists $\sigma_{\textrm{quar}}>0$ such that
\begin{equation*}
\min_{\Gamma= \gamma>0} 2\int \phi(f-\pi(f)) d\tilde{\pi} = \sigma_{\textrm{quar}}.
\end{equation*}
\end{prop}

The proof of Proposition~\ref{quarprop} can be found in Section~\ref{proofs}.

\section{Computation of the change in $\Gamma$}\label{nummeth}
Throughout this section, the $M = I_n$ case is considered. As mentioned, the formula~\eqref{funcdereq} gives a natural gradient descent direction~\eqref{DeltaGam} to take $\Gamma$ in order to optimise $\sigma^2$ from Theorem~\ref{theclt}. In Theorem~\ref{funcder} and in the form~\eqref{DeltaGam}, the expression for the gradient is already susceptible to a Green-Kubo approach in the sense that the form~\eqref{poissol} for $\phi$ can be substituted in to obtain a trajectory based formula, where finite difference is used to approximate $\nabla\!_p$ and independent realisations of $(q_t,p_t)$ is used for the expectations. However, this is too inaccurate in the implementation to be useful. The more directly calculable form as stated in the introduction in~\eqref{nabphi0} is used involving the derivative of $(q_t,p_t)$ with respect to the initial condition in Section~\ref{tangentproc}.\\
We focus the discussion on a Monte Carlo method to approximate $\nabla\phi$ and gradient directions in $\Gamma$ (e.g.~\eqref{DeltaGam}) based on Theorem~\ref{funcder}, but a spectral method to solve~\eqref{poisson0} and compute the change in $\Gamma$ is given in Appendix~\ref{num1}, which is computationally feasible in low dimensions. Algorithm~\ref{algo0} summarises the resulting continuous-time procedure, where all expectations within~\eqref{DeltaGam} are approximated by single realisations; further justifications, alternative methods, refinements and a concrete implementation (Algorithm~\ref{algorithm}) along with examples follow.

\begin{algorithm}[h]
\SetAlgoLined
\KwResult{$\Gamma\in\mathbb{S}_{++}^n$}
Start from arbitrary $(q_0,p_0)\in\mathbb{R}^{2n}$ and set $(\tilde{q}_0,\tilde{p}_0) = (q^0,-p^0)$, $D q_0 = D \tilde{q}_0 = 0$, $D p_0 = D \tilde{p}_0 =I_n$, $\zeta = \tilde{\zeta} = 0$, $k=0$, $\Gamma = I_n\quad \forall 1\leq j\leq B$, $t=t_0 = 0$\;
\For{$N$ epochs}{
simulate one step in $q_t$, $\tilde{q}_t$ then in $D_p q_t$ and $D_p \tilde{q}_t$ from $t$ to $t + \Delta t$\;
add to $\zeta,\tilde{\zeta}$ to approximate the row vectors \begin{equation*}\zeta = \int_{t_0}^t \nabla f(q_s)^\top D_p q_s ds,\quad \tilde{\zeta} = \int_{t_0}^t\nabla f(\tilde{q}_s)^\top D_p\tilde{q}_s ds;\end{equation*}
\If{$(D_p q_t,D_p p_t)$ is small enough in magnitude}{update $\Gamma$ with the gradient direction $-\zeta\otimes\tilde{\zeta}$\;
reset $(\tilde{q}_t,\tilde{p}_t) \leftarrow (q_t,-p_t)$; $(D_p q_t,D_p p_t), (D_p \tilde{q}_t,D_p \tilde{p}_t) \leftarrow (0,I_n)$; $t_0 \leftarrow t$; $\zeta,\tilde{\zeta} \leftarrow 0$\;
}
$t\leftarrow t+\Delta t$
}
\caption{Continuous-time outline of $\Gamma$ update using~\eqref{funcder00} and~\eqref{nabphi0}}\label{algo0}
\end{algorithm}

\subsection{Methodology}
Here we describe an on-the-fly procedure to repeatedly calculate the change~\eqref{DeltaGam} in $\Gamma$ by simulating the first variation process parallel to underdamped Langevin processes. The discretisation schemes used to simulate~\eqref{langevin00} and~\eqref{Dq1} are given in Section~\ref{discscheme}. Two gradient procedures, namely gradient descent and the Heavy ball method, for evolving $\Gamma$ given a gradient are detailed in Section~\ref{gradingam}. Then iterates from Section~\ref{discscheme} are used to approximate each change in $\Gamma$ in Section~\ref{thin} (see also Appendix~\ref{cigam1}. The key idea linking the above is that if equation~\eqref{newform} holds, then
\begin{align}
\Delta \Gamma &= \int \nabla\!_p\phi \otimes\nabla\!_p\tilde{\phi} d\tilde{\pi} \nonumber\\
& = -\int \int_0^\infty\mathbb{E}[\nabla f(q_s)^\top D_p q_s]^\top ds \int_0^\infty\mathbb{E}[\nabla f(\tilde{q}_t)^\top D_p \tilde{q}_t] dt d\tilde{\pi},\label{expecs}
\end{align}
where $(q_t,p_t)$ and $(\tilde{q}_t,\tilde{p}_t)$ denote the solutions to~\eqref{langevin00} with initial values $(q,p)$, $(q,-p)$ respectively, $(D_p q_t,D_p p_t)$ and $(D_p \tilde{q}_t,D_p \tilde{p}_t)$ denote the solutions to~\eqref{Dq} with $\tilde{q}_t$ replacing $q_t$ for the latter and the integral in~\eqref{expecs} is with respect to $(q,p)$. 

\subsubsection{Splitting}\label{discscheme}
A BAOAB splitting scheme \cite{MR3040887,doi:10.1063/1.4802990} will be used to integrate the Langevin dynamics~\eqref{langevin00}, given explicitly by
\begin{align}\label{langevindisc}
\begin{cases}
p^{i+\frac{1}{3}} &= p^i -\nabla U(q^i)\frac{\Delta t}{2} \\
q^{i+\frac{1}{2}} &= q^i + p^{i+\frac{1}{3}} \frac{\Delta t}{2}\\
p^{i+\frac{2}{3}} &= \exp({-\Delta t\Gamma^i})p^{i+\frac{1}{3}} + \sqrt{1-\exp({-2\Delta t\Gamma^i})} \xi^i\\
q^{i+1} &= q^{i+\frac{1}{2}} + p^{i+\frac{2}{3}} \frac{\Delta t}{2}\\
p^{i+1} &= p^{i+\frac{2}{3}} - \nabla U(q^{i+1})\frac{\Delta t}{2}
\end{cases}
\end{align}
for $i\in\mathbb{N}$, $\Delta t >0$, where $\xi^i$ are independent $n$-dimensional standard normal random variables and $\Gamma^i\in\mathbb{S}_{++}^n$ are a sequence of friction matrices to be updated throughout the duration of the algorithm, but we mention again recent developments, e.g. \cite{pmlr-v75-cheng18a, MR4091098, Foster2021TheSO, MR4309974, SanzSerna2021WassersteinDE, NEURIPS2019_eb86d510}, on discretisations of the underdamped Langevin dynamics; the majority of the numerical error involved in updating $\Gamma$ is expected to come from the small number of particles in approximating the integrals in the expression~\eqref{DeltaGam} for $\Delta\Gamma$, so that no further deliberation is made about the choice of discretisation for the purposes here. The first variation process~\eqref{Dq} together with its initial condition is 
\begin{subequations}\label{Dq}
\begin{align}
D_p q_t &= \int_0^t D_p p_s ds,\\
D_p p_t &= I_n-\int_0^t (D^2U(q_s)D_p q_s + \Gamma D_p p_s)ds.
\end{align}
\end{subequations}
In order to simulate~\eqref{Dq}, an analogous splitting scheme is used:
\begin{align}\label{Dqdisc}
\begin{cases}
Dp^{i+\frac{1}{3}} &= Dp^i - D^2 U(q^i)Dq^i\frac{\Delta t}{2} \\
Dq^{i+\frac{1}{2}} &= Dq^i + Dp^{i+\frac{1}{3}} \frac{\Delta t}{2}\\
Dp^{i+\frac{2}{3}} &= \exp({-\Delta t\Gamma^i})Dp^{i+\frac{1}{3}}\\
Dq^{i+1} &= Dq^{i+\frac{1}{2}} + Dp^{i+\frac{2}{3}} \frac{\Delta t}{2}\\
Dp^{i+1} &= Dp^{i+\frac{2}{3}} - D^2 U(q^{i+1})Dq^i\frac{\Delta t}{2}.
\end{cases}
\end{align}
The $k^{\textrm{th}}$ column of the first term including the Hessian of $U$ (and similarly for the last) can be approximated by 
\begin{equation}\label{d2appr}
-\nabla U(q^i + \frac{\Delta t}{2}(Dq^i)_k) + \nabla U(q^i) 
\end{equation}
where $(Dq^i)_k$ denotes the $k^{\textrm{th}}$ column of $Dq^i$, so that~\eqref{Dq} can still be approximated in the absence of Hessian evaluations. The approximation~\eqref{d2appr} will be used only when explicitly stated in the sequel.

\subsubsection{Gradient procedure in $\Gamma$}\label{gradingam}

Suppose we have available a series of proposal updates $(b_0,\dots,b_{L-1})\in\mathbb{R}^{n\times n \times L}$ for $\Gamma$. Given stepsizes $\alpha^i = \alpha\in\mathbb{R}$ and an annealing factor $r\in\mathbb{R}$, the following constrained stochastic gradient descent (for $i$ where proposal updates are produced)
\begin{equation}\label{Gamdisc1}
\Gamma^{i+1} = \Pi_{\textrm{pd}}^\mu\bigg(\Gamma^i  + \frac{\alpha^i}{2L}\sum_{j = 0}^{L-1} (b_j+b_j^\top)\bigg)
\end{equation}
can be considered, where $L\in\mathbb{N}$ and $\Pi_{\textrm{pd}}^\mu$ is the projection to a positive definite matrix, for some minimum value $\mu>0$, given by
\begin{equation}\label{muproj}
\Pi_{\textrm{pd}}^\mu(M) = \sum_{i = 1}^n \max(\lambda_i,\mu) v_i v_i^\top
\end{equation}
for symmetric $M\in\mathbb{R}^{n\times n}$ and its the eigenvalue decomposition
\begin{equation*}
M = \sum_{i = 1}^n \lambda_i v_i v_t^\top.
\end{equation*}
Alternatively, a Heavy-ball method \cite{POLYAK19641,7330562} (with projection) can be used. The method is considered in the stochastic gradient context in \cite{545bd4035e2245bcb16f18f063b0cf76}, given here as
\begin{align}\label{Gamdisc2}
\begin{cases}
\Gamma^{i+1} &= \Pi_{\textrm{pd}}^\mu(\Gamma^i  + \alpha^i \Theta^{i+1}),\\
\Theta^{i+1} &= (1-\alpha^i r)\Theta^i + \frac{\alpha^i}{2L}\sum_{j = 0}^{L-1} (b_j+b_j^\top).
\end{cases}
\end{align}
The heavy-ball method offers a smoother trajectory of $\Gamma$ over the course of the algorithm. 
Under appropriate assumptions on $b_j$, in particular if 
\begin{equation*}
\frac{1}{2L}\sum_{j = 0}^{L-1} (b_j+b_j^\top) \sim \mathcal{N}(\nabla \sigma^2(\Gamma^i), \sigma_b^2 I_{n^2}),
\end{equation*}
for some gradient $\nabla \sigma^2(\Gamma^i_k)$ and variance $\sigma_b^2>0$, then the system~\eqref{Gamdisc2} has the interpretation of an Euler discretisation of a constrained Langevin dynamics, in which case $\frac{r}{\sqrt{\alpha^i}\sigma_b^2}$ is the inverse temperature. By increasing $r$, the analogous invariant distribution `sharpens' around the maximum in its density and in this way reduces the effect of noise at equilibrium; on the other hand, decreasing $r$ reduces the decay in the momentum.

\subsubsection{A thinning approach for $\Delta\Gamma$}\label{thin}
The most straightforward way of approximating the integral in~\eqref{expecs} is to use independent realisations of~\eqref{langevindisc}, as described in Appendix~\ref{cigam1}, but we draw alternatively a thinned sample \cite{MR3698682} from a single trajectory here in order to run only a single parallel set of realisations of~\eqref{langevindisc} and~\eqref{Dqdisc} at a time. 
More specifically, we consider a single realisation of~\eqref{langevindisc} and regularly-spaced points from its trajectory (possibly after a burn-in) as sample points from $\tilde{\pi}$. Starting at each of these sample points and ending at each subsequent one, the process is replicated albeit starting with a momentum reversal and simulated in parallel. In addition, for each of the two processes, a corresponding first variation process~\eqref{Dqdisc} is calculated in parallel. A precise description follows.


Let $K=1$ for simplicity. The $\Gamma$ direction~\eqref{expecs} is approximated by
\begin{align}
-\frac{1}{(L+L^*)}\sum_{l=0}^{L+L^*-1}& \bigg( \sum_{i=1}^T  \frac{\Delta t}{K}\sum_{k = 1}^K \nabla f(q_{(k)}^{i+Tl+B})^\top Dq_{(k)}^{i+Tl+B} \bigg)\otimes\nonumber\\
& \bigg(  \sum_{i=1}^T \frac{\Delta t}{K}\sum_{k = 1}^K \nabla f(\tilde{q}_{(k)}^{i+Tl+B})^\top D\tilde{q}_{(k)}^{i+Tl+B} \bigg),\label{thinap}
\end{align}
where $L\in\mathbb{N}$, $((q_{(k)}^i,p_{(k)}^i))_{i\in\mathbb{N}}$, $((\tilde{q}_{(k)}^i,\tilde{p}_{(k)}^i))_{i\in\mathbb{N}}$ denote solutions to~\eqref{langevindisc}
\begin{itemize}
\item for $i\neq B+Tl-1$, $l\in\mathbb{N}$ if $k\neq 1$ and
\item for all $i$ if $k=1$
\end{itemize}
with initial condition $(0,0)$, noise $\xi^i = \xi_{(k)}^i,\tilde{\xi}_{(k)}^i$ for all $i\in\mathbb{N}$ satisfying $\xi_{(k)}^i = \xi_{(k')}^i = \tilde{\xi}_{(k)}^i = \tilde{\xi}_{(k')}^i$ for all $i < B, 1\leq k\leq K$, $1\leq k' \leq K$, independent otherwise as $i$ and $k$ vary, along with corresponding $(Dq_{(k)}^i,Dp_{(k)}^i)$, $(D\tilde{q}_{(k)}^i,D\tilde{p}_{(k)}^i)$ satisfying~\eqref{Dqdisc} for $i\neq B+Tl-1$, $l\in\mathbb{N}$ (regardless of $k$), and where the $k\neq 1$ processes are `reset' at $i = B+Tl$ corresponding to the values of the $k=1$ chain if the first variation processes have converged to zero, that is,
\begin{subequations}\label{resetv}
\begin{align}
q_{(k)}^{Tl+B} &= q_{(1)}^{Kl+B}, \qquad p_{(k)}^{Tl+B} = p_{(1)}^{Tl+B}, \qquad Dq_{(k)}^{Tl+B} = 0,\qquad Dp_{(k)} ^{Tl+B} = I_n\\
\tilde{q}_{(k)}^{Tl+B} &= q_{(1)}^{Tl+B}, \qquad \tilde{p}_{(k)}^{Tl+B} = -p_{(1)}^{Tl+B}, \qquad D\tilde{q}_{(k)}^{Tl+B} = 0,\qquad D\tilde{p}_{(k)} ^{Tl+B} = I_n
\end{align}
\end{subequations}
for all $1\leq k\leq K$ if for some $D_{\textrm{conv}}>0$,
\begin{subequations}\label{convconds}
\begin{align}
\max_{i,j,k}\abs*{(Dq_{(k)}^{Tl+B})_{ij}} < D_{\textrm{conv}}, &\quad \max_{i,j,k}\abs*{(D\tilde{q}_{(k)}^{Tl+B})_{ij}} < D_{\textrm{conv}},\\
\quad \max_{i,j,k}\abs*{(Dp_{(k)}^{Tl+B})_{ij}} < D_{\textrm{conv}}, &\quad \max_{i,j,k}\abs*{(D\tilde{p}_{(k)}^{Tl+B})_{ij}} < D_{\textrm{conv}}
\end{align}
\end{subequations}
and $L^*\in\mathbb{N}$ is such that the number of elements in $\{ l\in\mathbb{N} : 1\leq l\leq L+L^*\}$ satisfying~\eqref{convconds} is $L$. The approach is summarised in Algorithm~\ref{algorithm}. Of course, the above for generic $K\in\mathbb{N}$ constitutes improving approximations to $\Delta \Gamma$. Note that as $\Gamma$ changes through the prescribed procedure, the asymptotic variance associated to the given observable $f$ is expected to improve, but on the contrary, the estimator~\eqref{thinap} for the continuous-time expression~\eqref{expecs} may well worsen, since the integrand (of the outermost integral) in~\eqref{expecs} is not $f$. Increasing $L$ is expected to solve any resulting issues; extremely small $L$ have been successful in the experiments here.

\begin{algorithm}[h]
\SetAlgoLined
\KwResult{$\Gamma^i, 1\leq i \leq N+1$}
Start from arbitrary $(q^0,p^0)\in\mathbb{R}^{2n}$ and set $D q^0 = D \tilde{q}^0 = 0$, $D p^0 = D \tilde{p}^0 =I_n$, $\zeta = \tilde{\zeta} = 0$, $k=0$, $\Gamma^j = I_n\quad \forall 1\leq j\leq B$\; 
\For{$i = 1:B-1$}{compute $q^{i+1}$ according to~\eqref{langevindisc}\;}
\If{$ i = B$}{set $(\tilde{q}^i,\tilde{p}^i) \leftarrow (q^i,-p^i)$\;}
\For{$i = B:N$}{
compute $q^{i+1}$ and $\tilde{q}^{i+1}$ according to~\eqref{langevindisc}\;
compute $D q^{i+1}$ and $D \tilde{q}^{i+1}$ from~\eqref{Dqdisc} corresponding to $q^{i+1}$ and $\tilde{q}^{i+1}$ respectively\;
compute the row vectors $\quad\begin{aligned}[t]\zeta &\leftarrow \zeta + \nabla f(q^{i+1})^\top Dq^{i+1} \Delta t\\ \tilde{\zeta} &\leftarrow \tilde{\zeta} + \nabla f(\tilde{q}^{i+1})^\top D\tilde{q}^{i+1} \Delta t\end{aligned}$\;
\If{$l:=i-B \in T\mathbb{N}$ and~\eqref{convconds} hold (ignoring appearances of $(k) = (1)$)}{save the matrix $ b_{(\frac{k}{G} - \floor{\frac{k}{G}})G} = -\zeta \otimes \tilde{\zeta}$\;
reset as follows: $\begin{aligned}[t]\zeta &, \tilde{\zeta} \leftarrow 0, & (\tilde{q}^{i+1},\tilde{p}^{i+1}) &\leftarrow (q^{i+1},-p^{i+1})\\
D q^{i+1} &, D\tilde{q}^{i+1} \leftarrow 0,& D p^{i+1} &, D \tilde{p}^{i+1} \leftarrow I_n\end{aligned}$\;
and update the counter $k \leftarrow k+1$\;
}
\eIf{$k \in G\mathbb{N}$}{
compute $\Gamma^{i+1}$ according to~\eqref{Gamdisc2}\;} 
{set $\Gamma^{i+1} = \Gamma^i$\;}
}
\caption{Gradient procedure in $\Gamma$}\label{algorithm}
\end{algorithm}

\begin{remark}\label{simul}
If it is of interest to approximate expectations of $P\in\mathbb{N}$ observables with respect to $\pi$, the quantity $\sum_i^P\sigma_i^2$ for example can be used as an objective function, where $\sigma_i^2$ is the asymptotic variance from the $i^{\textrm{th}}$ observable. In the implementation in Algorithm~\ref{algorithm}, instead of only the vectors $\zeta$, $\tilde{\zeta}$, this amounts to calculating at each iteration the vectors $\zeta^{(i)}$, $\tilde{\zeta}^{(i)}$ corresponding to the $i^{\textrm{th}}$ observable and taking the sum of the resulting update matrices in $\Gamma$ to update $\Gamma$. This calls for no extra evaluations of $\nabla U$ over the single observable case.
\end{remark}


\begin{remark}\textit{(Tangent processes along random directions)}\label{alodir}
We mention the situation where simulating the full first variation processes $(D_p q_t,D_p p_t)$ in $\mathbb{R}^{n\times 2n}$ is prohibitively expensive. A directional tangent process can be used instead of $(D_p q_t, D_p p_t)$. Consider for a unit vector $v\in\mathbb{R}^n$, that is $\abs*{v}=1$, randomly chosen at the beginning of each estimation of $\Delta\Gamma$, the pair of vectors $(D_p q_t v, D_p p_t v)\in\mathbb{R}^{n\times 2}$. Multiplying on the right of (\eqref{Dq} and)~\eqref{Dqdisc} by $v$, one obtains
\begin{align}\label{Dqvdisc}
\begin{cases}
Dpv^{i+\frac{1}{3}} &= Dpv^i - D^2 U(q^i)Dqv^i\frac{\Delta t}{2} \\
Dqv^{i+\frac{1}{2}} &= Dqv^i + Dpv^{i+\frac{1}{3}} \frac{\Delta t}{2}\\
Dpv^{i+\frac{2}{3}} &= \exp({-\Delta t\Gamma^i})Dpv^{i+\frac{1}{3}}\\
Dqv^{i+1} &= Dqv^{i+\frac{1}{2}} + Dpv^{i+\frac{2}{3}} \frac{\Delta t}{2}\\
Dpv^{i+1} &= Dpv^{i+\frac{2}{3}} - D^2 U(q^{i+1})Dqv^i\frac{\Delta t}{2},
\end{cases}
\end{align}
where the first term involving the Hessian of $U$ in~\eqref{Dqvdisc} can be approximated by 
\begin{equation*}
-\nabla U(q^i + \frac{\Delta t}{2}Dqv^i) + \nabla U(q^i) 
\end{equation*}
and similarly for the last such term. In continuous time, the resulting direction in $\Gamma$ is $\int\nabla\phi^\top v \nabla\tilde{\phi}^\top v d\tilde{\pi} v\otimes v$ and from~\eqref{funcdereq} the rate of change in asymptotic variance in this direction is $-2(\int \nabla\phi^\top v \nabla\tilde{\phi}^\top v d\tilde{\pi})^2$. 
However, the resulting gradient procedure in $\Gamma$ turns out to be painstakingly slow to converge in high dimensions in comparison to simulating a full first variation process; as illustration, one can think of the situation where the randomly chosen vector $v$ is taken from the restricted set of standard Euclidean basis vectors, where only one diagonal value in $\Gamma$ is changed at a time. 
For a directional derivative, we also mention \cite{MR2269681, MR2917400}.
\end{remark}

\subsection{Concrete examples}\label{conc}
In Sections~\ref{1dex},~\ref{diffbrid} and~\ref{bayinf}, the Monte Carlo approach is applied on concrete problems. Section~\ref{1dex} contains the simplest one-dimensional Gaussian case where the optimal $\Gamma$ is known and it is shown that the algorithm approximates it quickly. A different Gaussian problem extracted from a diffusion bridge context is explored in Section~\ref{diffbrid}, where the algorithm is shown to approximate a $\Gamma$ matrix that exhibits an even better empirical asymptotic variance than the one given by Proposition~\ref{comprop}. Finally, the algorithm is applied to finding the optimal $\Gamma$ in estimating the posterior mean in a Bayesian inference problem in Section~\ref{bayinf}, where the situation is shown to be similar to Proposition~\ref{lincor}, in the sense that the optimal $\Gamma$ is close to $0$; after and separately from such a finding, the empircal asymptotic variance for a small $\Gamma$ is compared that for $\Gamma = I_n$, with dramatic improvement in both the full gradient and minibatch gradient cases.

\subsubsection{One dimensional quadratic case}\label{1dex}
Here the algorithm given in Section~\ref{thin} is used in the simplest one dimensional
\begin{equation}\label{quadob1d}
U(q) = \frac{V_0}{2}q^2,\qquad f(q) = \frac{1}{2}q^2,
\end{equation}
$V_0>0$, case to find the optimal constant friction. Since commutativity issues disappear in the one-dimensional case, the optimal constant friction is known analytically and is given by Proposition~\ref{comprop} to be $\Gamma = \sqrt{V_0}$, with the asymptotic variance $V_0^{-\frac{5}{2}}$. Moreover, the relationship between the asymptotic variance and $\Gamma$ is explicitly given by equations~\eqref{gexp} and~\eqref{asympfor}, which reduces in this case to
\begin{equation*}
\sigma^2(\Gamma) = \frac{1}{4V_0^2}\bigg(\Gamma^{-1} + \frac{1}{V_0}\Gamma\bigg).
\end{equation*}
The case $V_0=5$ is illustrated in Figure~\ref{asfig0}. In the middle and right plot of Figure~\ref{asfig0}, the procedure in Section~\ref{thin} is used for $5\cdot 10^4$ epochs, with $\Delta t = 0.08$, block-size $T = 125$, $L = 1$ and $D_{\textrm{conv}} = 2\cdot 10^{-4}$. Changing the observable to the linear
\begin{equation}\label{linob1d}
f(q) = q
\end{equation}
gives that the `optimal' (but unreachable in the algorithm due to the constraints) friction is $0$ by Corollary~\ref{lincor}. The right plot in Figure~\ref{asfig0} shows that the procedure arrives at a similar conclusion in the sense that the $\Gamma$ hits and stays at $\mu = 0.2$.

\begin{figure}[h]
  \centering
  \begin{subfigure}[h]{0.31\linewidth}
    \adjincludegraphics[trim={6mm 0 3mm 0},width=\linewidth]{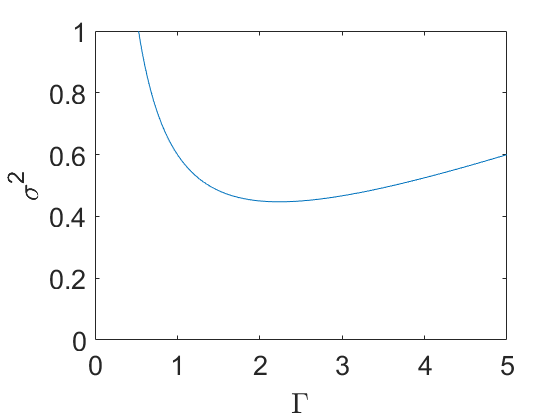}
  \end{subfigure}
  \begin{subfigure}[h]{0.31\linewidth}
    \adjincludegraphics[trim={6mm 0 3mm 0},width=\linewidth]{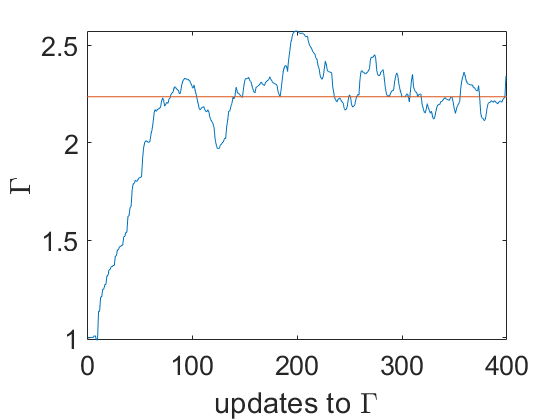}
  \end{subfigure}
  \begin{subfigure}[h]{0.31\linewidth}
    \adjincludegraphics[trim={6mm 0 3mm 0},width=\linewidth]{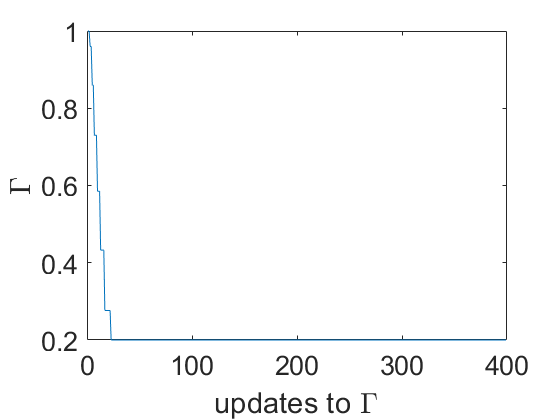}
  \end{subfigure}
  
  \caption{Left: Relationship between asymptotic variance and $\Gamma$ for~\eqref{quadob1d}. Middle and right: Trajectory of $\Gamma$ for~\eqref{quadob1d} and~\eqref{linob1d} respectively by~\eqref{Gamdisc2} with $\alpha^i = 1$, $G=1$, $r = 0.5$ and $\mu = 0.2$. Middle: the red line is the optimal value $\Gamma = \sqrt{5}$ given by Proposition~\ref{comprop}. All plots are for $V_0 = 5$. }
  \label{asfig0}
\end{figure}

\subsubsection{Diffusion bridge sampling}\label{diffbrid}
The algorithm in Section~\ref{thin} is applied in the context of diffusion bridge sampling \cite{MR2358638,MR2188686} (see also for example \cite{MR2588600,MR2269221,MR2807970}), where the SDE
\begin{equation}\label{bridge1}
dx_t = -\nabla V(x_t)dt + \sqrt{2\beta^{-1}}dW_t'
\end{equation}
for a suitable $V:\mathbb{R}^d\rightarrow\mathbb{R}$, $\beta >0$ and $W_t'$ standard Wiener process on $\mathbb{R}^d$, is conditioned on the events 
\begin{equation}\label{bridge2}
x_0 = x_-\quad\textrm{and}\quad x_T = x_+
\end{equation}
for some fixed $T>0,x_0,x_+ \in\mathbb{R}^d$ and the problem setting is to sample from the path space of solutions to~\eqref{bridge1} conditioned on~\eqref{bridge2}. For the derivation of the following formulations, we refer to Section 5 in \cite{MR2358638} and Section 6.1 in \cite{MR2444507}; here we extract a simplified potential $U$ to apply our algorithm on after a brief description.\\
Let
\begin{equation*}
V(x) = \frac{1}{2}\abs*{x}^2,\qquad x_- = x_+ = 0,\qquad \beta = 1, \qquad d=1,\qquad T=1.
\end{equation*}
Using the measure given by Brownian motion conditioned on~\eqref{bridge2} as the reference measure $\mu_0$ on the path space of continuous functions $C([0,1],\mathbb{R})$, the measure $\mu$ associated to~\eqref{bridge1} conditioned on~\eqref{bridge2} satisfies
\begin{equation*}
\frac{d\mu}{d\mu_0}(x) \propto \exp\bigg(-\frac{1}{4}\int_0^T \abs*{x}^2 dt \bigg),
\end{equation*}
where the left hand side denotes the Radon-Nikodym derivative, so that discretising $\mu$ on a grid in $[0,1]$ with grid-size $\delta>0$ gives the approximating measure
\begin{equation*}
\pi(q_1,\dots,q_n) \propto e^{-U(q_1,\dots,q_n)}
\end{equation*}
where $U$ is given by 
\begin{equation*}
U(q) = \frac{1}{2}q^\top \Sigma^{-1} q = \frac{1}{2}q^\top\begin{pmatrix}
\frac{2}{\delta} + \frac{\delta}{4} & -\frac{1}{\delta} & & &\\
-\frac{1}{\delta} & \frac{2}{\delta} + \frac{\delta}{4} & -\frac{1}{\delta} & &\\
&\dots & & &\\
& & -\frac{1}{\delta} & \frac{2}{\delta} + \frac{\delta}{4} & -\frac{1}{\delta} \\
& & & \frac{2}{\delta} + \frac{\delta}{4} & -\frac{1}{\delta}
\end{pmatrix}q.
\end{equation*}
From here the Langevin system~\eqref{langevin00} can be used to sample from $\pi$ and the algorithm given in Section~\ref{thin} is applied. For this purpose, the observable
\begin{equation*}
f(q) = \frac{1}{2}\abs*{q}^2
\end{equation*}
is used together with the parameters $\delta = \frac{1}{21}$, $n= 20$, $K=1$, $L = 5$, $T=60$, $B = 100$ and $D_{\textrm{conv}} = 0.01$. Only the diagonal values of $\Gamma$ are updated and their trajectories are shown in Figure~\ref{diagdiffb}.

\begin{figure}[h]
  \centering
    \includegraphics[width=0.5\linewidth]{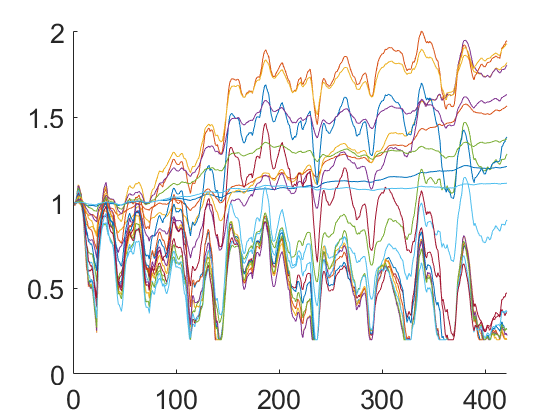}
\caption{Diagonal values of $\Gamma$ over iterations of~\eqref{Gamdisc2} with $\alpha^i = 0.2$, $G = 5$, $r = 1$ and $\mu=0.2$.}\label{diagdiffb}
  \end{figure}
At the end of $300000$ epochs, $\Gamma$ is given by
\begin{align*}
\Gamma_{\textrm{final}} =\textrm{diag}(&  1.2129, 1.5673, 1.8199, 1.8055, 1.2858, 0.9013, 0.3588, 0.2631,\\
& 0.2000, 0.2000, 0.2252, 0.2579, 0.3621, 0.4715, 1.3842, 1.9467,\\
& 1.9289, 1.6326, 1.3730, 1.1153).
\end{align*}
This $\Gamma$ is fixed and used for a standard sampling procedure for the same potential and observable. The asymptotic variance is approximated by grouping the epochs after $B = 100$ burn-in iterations into $N_B = 999$ blocks of $T = 300$ epochs, specifically, 
\begin{equation*}
\sigma_{\textrm{approx}} = \frac{1}{N_B}\sum_{l = 0}^{N_B-1} \bigg[ \frac{1}{\sqrt{T\Delta t}}\sum_{i=1}^T \Delta t\bigg( f(q^{i+Tl+B})  - \frac{1}{N}\sum_{j=1}^{N} \Delta t f(q^{j+B}) \bigg)\bigg]^2
\end{equation*}
and this is compared to the estimate from the same procedure using different values of fixed $\Gamma$ in Table~\ref{table0}. Note that $\Gamma = \Sigma^{-\frac{1}{2}}$ is the optimal $\Gamma$ in the restricted class of matrices commuting with $\Sigma$ given by Proposition~\ref{comprop}, where the asymptotic variance is known to be $\textrm{Tr}(\Sigma^{\frac{5}{2}}) \approx 6.4785$.
\begin{table}[h!]
\begin{center}
\begin{tabular}{ | c | c | c | } 
\hline
 & $\sigma_{\textrm{approx}}$ \\ 
\hline
$\Gamma = I_n$ & $6.9834$\\ 
\hline
$\Gamma = \Sigma^{-\frac{1}{2}}$ & $6.5096$\\
\hline
$\Gamma = \Gamma_{\textrm{final}}$ & $6.1667$\\ 
\hline
\end{tabular}
\caption{Empirical asymptotic variances with $N_B = 999$, $T=300$, $B=100$, $N = 299700$. }
\label{table0}
\end{center}
\end{table}

\subsubsection{Bayesian inference}\label{bayinf}
We adopt the binary regression problem as in \cite{durmus2018highdimensional} on a dataset\footnote{http://archive.ics.uci.edu/ml/datasets/Internet+Advertisements. Note that besides missing values at some datapoints, the dataset comes with many quantitatively duplicate features and also some linear dependence between the vectors made up of a single feature across all datapoints; here features have been removed so that the said vectors remaining are linearly independently. In particular, $n=642$.} with datapoints encoding information about images on a webpage and each labelled with `ad' or `non-ad'. The labels $\{ Y_i \}_{1\leq i\leq p}$, taking values in $\{0,1\}$, of the $p=2359$ datapoints (counting only those without missing values) given in the dataset are modelled as conditionally independent Bernoulli random variables with probability $\{ \rho(\beta^\top X_i) \}_{1\leq i \leq p}$, where $\rho$ is the logistic function given by $\rho(z) = e^{cz}/(1+e^{cz})$ for all $z\in\mathbb{R}$, $c\in\mathbb{R}$ is given by~\eqref{clab1}, $\{ X_i \}_{1\leq i \leq p}$, $\beta$, both taking values in $\mathbb{R}^n$, are respectively vectors of known features from each datapoint and regression parameters to be determined. The parameters $\beta$ are given the prior distribution $\mathcal{N}(0,\Sigma)$, where 
\begin{equation*}
\Sigma^{-1}=\frac{1}{p}\sum_{i=1}^p X_i^\top X_i\in\mathbb{R}^{n\times n},
\end{equation*}
and the density of the posterior distribution of $\beta$ is given up to proportionality by
\begin{equation*}
\pi_{\beta}(\beta|\{ (X_i,Y_i) \}_{1\leq i\leq p}) \propto \exp\bigg(\sum_{i=1}^p \{cY_i\beta^\top X_i - \log(1+e^{c\beta^\top X_i})\} - \frac{1}{2}\beta^\top \Sigma^{-1}\beta\bigg),
\end{equation*}
so that the log-density gradient, in our notation $-\nabla U$, is given by
\begin{equation*}
-\nabla U(\beta) = \sum_{i=1}^p  cX_i (Y_i - (1+e^{-c\beta^\top X_i})^{-1}) -\Sigma^{-1} \beta.
\end{equation*}
The observable vector $f_i(q) = q_i$, $1\leq i \leq n$, corresponding to the posterior mean is used. The coordinate transform $\hat{\beta} = \Sigma^{-\frac{1}{2}}\beta$ is made before applying the symmetric preconditioner $\Sigma^{\frac{1}{2}}$ on the Hamiltonian part of the dynamics so that the dynamics simulated are as in~\eqref{langevin00} with $M=I_n$ and 
\begin{equation}\label{newgrad}
-\nabla U(\hat{\beta}) = \Sigma^{\frac{1}{2}}\sum_{i=1}^p  cX_i (Y_i - (1+e^{-c(\Sigma^{\frac{1}{2}}\hat{\beta})^\top X_i})^{-1}) - \hat{\beta}.
\end{equation}
We use the observable vector
\begin{equation*}
f_i(\hat{\beta}) = \hat{\beta}_i,\quad 1\leq i \leq n
\end{equation*}
and the sum of their corresponding asymptotic variances as the value to optimise with respect to $\Gamma$, but show in Figures~\ref{asfig0} and~\ref{asfig} the estimated asymptotic variances for both sets $f_i(\hat{\beta})$, $f_i(\beta)$ of observables, where the estimation is calculated using the vector on the left of the outer product in~\eqref{thinap} in accordance with $2\int \nabla\phi^\top\Gamma\nabla\phi d\tilde{\pi}$ which follows from the formula~\eqref{cltform} after integrating by parts with truncation. The approximation~\eqref{d2appr} for the term(s) including the Hessian in~\eqref{Dqdisc} has been used to test the method despite the explicit availability of the Hessian. During the execution of Algorithm~\ref{algorithm}, the constant $c$ has been set to
\begin{equation}\label{clab1}
c = \bar{c} := \frac{5}{\max_i(\Sigma^{\frac{1}{2}}\sum_j X_j Y_j)_i}.
\end{equation}
In detail, 30000 epochs are simulated; after 100 burn-in iterations of the Langevin discretisation~\eqref{langevindisc}, 2 parallel simulations of~\eqref{langevindisc} and 2 of the first variation discretisation~\eqref{Dqdisc} are run according to Section~\ref{thin} with time-step $\Delta t = 0.1$, block-size $T=100$, $L=1$ block per update in $\Gamma$, $K=1$ and tolerance $D_{\textrm{conv}} = 0.01$. 

\begin{figure}[h]
  \centering
  \begin{subfigure}[h]{0.31\linewidth}
    \adjincludegraphics[trim={8mm 0 5mm 0},width=\linewidth]{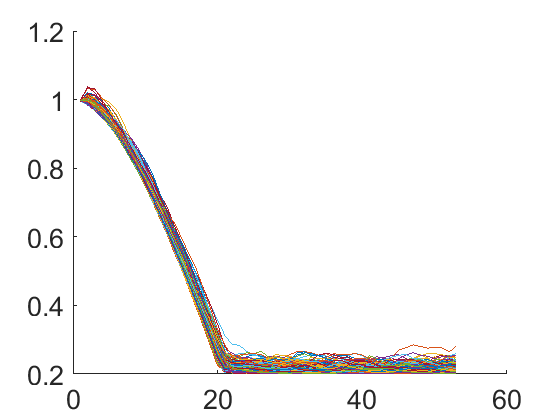}
  \end{subfigure}
  \begin{subfigure}[h]{0.31\linewidth}
    \adjincludegraphics[trim={8mm 0 5mm 0},width=\linewidth]{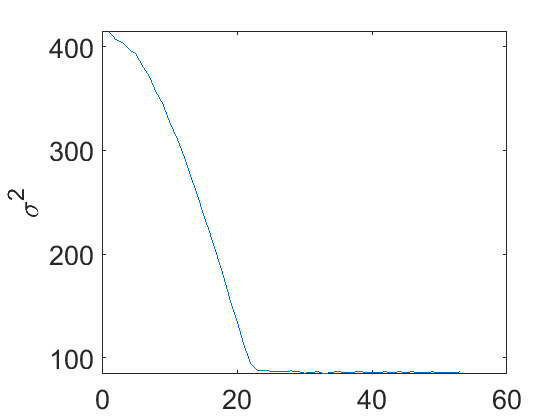}
  \end{subfigure}
   \begin{subfigure}[h]{0.31\linewidth}
    \adjincludegraphics[trim={8mm 0 6mm 0},width=\linewidth]{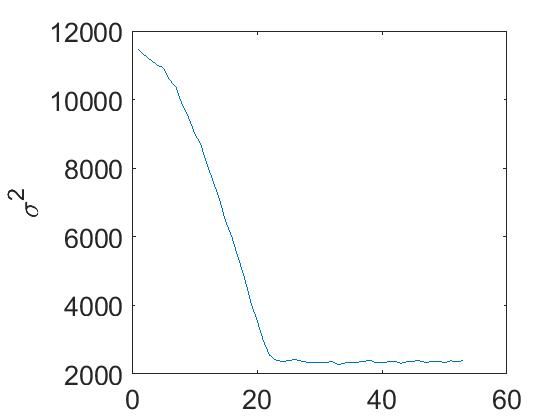}
  \end{subfigure}
  
  \caption{Left: Diagonal values of $\Gamma$ over iterations of~\eqref{Gamdisc2} with $\alpha^i = 0.1$, $G=1$, $r = 1$ and $\mu=0.2$. Note that the mean of the absolute values of all entries of $\Gamma$ at the end of the iterations is $0.0039$. Middle: Sum over $i$ of estimated asymptotic variances for $f_i(\hat{\beta})$; right: for $f_i(\beta)$. }
  \label{asfig0}
\end{figure}

\begin{figure}[h]
  \centering
  \begin{subfigure}[h]{0.31\linewidth}
    \adjincludegraphics[trim={8mm 0 5mm 0},width=\linewidth]{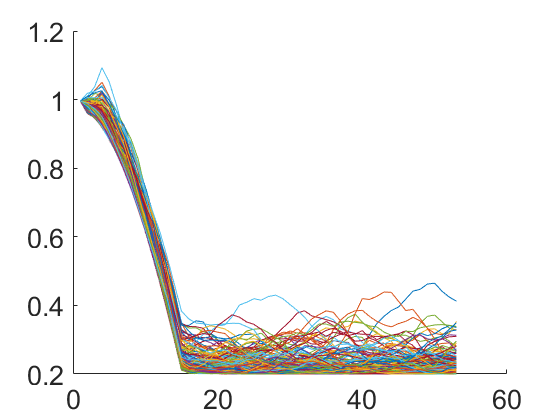}
  \end{subfigure}
  \begin{subfigure}[h]{0.31\linewidth}
    \adjincludegraphics[trim={8mm 0 5mm 0},width=\linewidth]{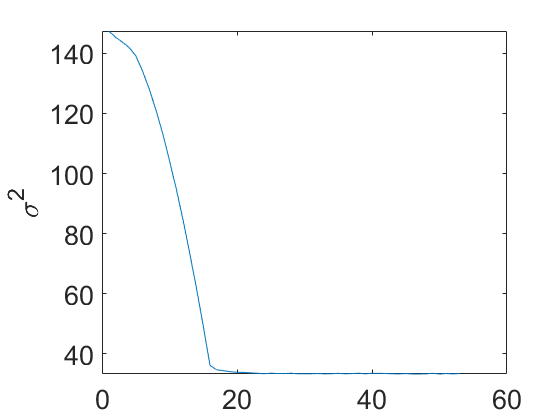}
  \end{subfigure}
   \begin{subfigure}[h]{0.31\linewidth}
    \adjincludegraphics[trim={8mm 0 6mm 0},width=\linewidth]{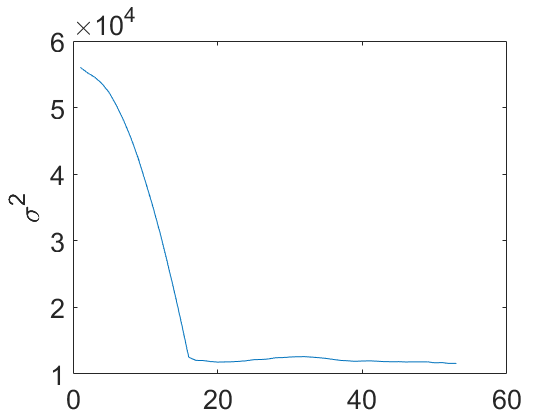}
  \end{subfigure}
  
  \caption{The same as in the caption of Figure~\ref{asfig}, except $r=0.5$ and a different dataset (https://archive.ics.uci.edu/ml/datasets/Musk+(Version+1)) is used where $n = 167$ and $p = 476$. The mean of the absolute values of all entries of $\Gamma$ at the end of the iterations is $0.0210$.}
  \label{asfig}
\end{figure}

In Figures~\ref{asfig0} and~\ref{asfig}, $\Gamma$ starts initially from the identity $I_n$ and descends towards $0.2I_n$ (restricted as in~\eqref{muproj}), as expected for a linear observable and potential close to a quadratic (see Proposition~\ref{oddprop}). We note that in the gradient descent procedure for $\Gamma$, using the minibatch gradient does not change the behaviour shown in Figures~\ref{asfig0} and~\ref{asfig}. In addition, although the trajectory of $\Gamma$ seems to go directly to zero, we expect the optimal $\Gamma$ to be close but away from zero since the potential is close but not exactly quadratic.\\
Next, the value for $\Gamma$ is fixed at various values and used for hyperparameter training on the same problem for the first dataset, using both the full gradient~\eqref{newgrad} and a minibatch\footnote{The control variate stochastic gradient on underdamped dynamics \cite{Chatterji2018OnTT,MR4227705} is not directly considered here but the benefits of an improved $\Gamma$ is expected to carry over to such variations of the stochastic gradient.} version where the sum in~\eqref{newgrad} is replaced by $\frac{p}{10}$ times a sum over
a subset $S$ of $\{1,\dots,p\}$ with $10$ elements randomly drawn without replacement such that $S$ changes once for each $i$ in~\eqref{langevindisc}. In the minibatch gradient case, $c$ is set to a fraction of~\eqref{clab1}, specifically $\bar{c}(\frac{p}{10})^{-1}$. In Tables~\ref{table1} and~\ref{table2}, variances for the posterior mean estimates are shown (similar variance reduction results persist when using the probability of success for features taken from a single datapoint in the dataset). \\
In detail, for each row of Tables~\ref{table1} and~\ref{table2}, $N = 29700$ epochs of~\eqref{langevindisc} are simulated with the same parameters as above. The asymptotic variance for each observable entry is approximated using block averaging (Section 2.3.1.3 in \cite{MR2681239}) by grouping the epochs after $B = 100$ burn-in iterations into $N_B = 99$ blocks of $T = 300$ epochs, that is,
\begin{equation*}
\sigma_{k,\textrm{approx}}^2 = \frac{1}{N_B}\sum_{l = 0}^{N_B-1} \bigg[ \frac{1}{\sqrt{T\Delta t}}\sum_{i=1}^T \Delta t\bigg( f_k(q^{i+Tl+B})  - \frac{1}{N}\sum_{j=1}^{N} f_k(q^{j+B}) \bigg)\bigg]^2
\end{equation*}
and $N_B = 3$ blocks of $T= 9900$ epochs (respectively for each column of Tables~\ref{table1} and~\ref{table2}); the values $0.8667$ and $0.1571$ approach and correspond to values in the middle plot of Figure~\ref{asfig0} after multiplying by $n=642$.  The variances are compared to those using a gradient oracle: unadjusted (overdamped) Langevin dynamics\cite{durmus2018highdimensional} and with an irreversible perturbation\cite{MR3483241}, where the antisymmetric matrix $J$ is given by
\begin{equation*}
J_{i,j} = \begin{cases}
1 & \textrm{if } j - i  = 1\textrm{ or } 1-n,\\
-1 & \textrm{if } i - j = 1\textrm{ or } 1-n,\\
0 & \textrm{otherwise}
\end{cases}
\end{equation*}
for $1\leq i,j\leq n$ and the stepsizes are the same as for underdamped implementations. In addition, the Euclidean distance from intermediate estimates of the posterior mean to a total, combined estimate is shown for each method. Specifically, $d_k := \abs*{\frac{1}{300k}\sum_{i = 1}^{300k} f(q^{i+B}) - \hat{\pi}(f)}$ is plotted against $k$ in Figure~\ref{intermeans}, where $\hat{\pi}(f)$ is the mean (over the methods listed in Tables~\ref{table1} and~\ref{table2}) of the final posterior mean estimates. A weighted mean with unit weights except one half on the $\Gamma=0.2I_n$ and $\Gamma=0.1I_n$ methods also gave similar results, though this is not shown explicitly.
\begin{table}[h!]
\begin{center}
\begin{tabular}{ | c | c | c | } 
\hline
 & block-size $T=300$  & block-size $T=9900$ \\ 
\hline
$\Gamma = I_n$ & (1.2669,0.0320) & (0.8667,0.7190)\\ 
\hline
$\Gamma = 0.2I_n$ & (0.2939,0.0018) & (0.1571,0.0243)\\ 
\hline
$\Gamma = 0.1I_n$ & (0.1739,0.0007) & (0.0890,0.0092)\\
\hline
overdamped & (1.2298,0.0319) & (0.8687,0.8662)\\
\hline
irreversible overdamped & (0.5642,0.0077) & (0.3835,0.1614)\\
\hline
\end{tabular}
\caption{$(\frac{1}{n}\sum_{k = 1}^n \sigma_{k,\textrm{approx}}^2, \frac{1}{n}\sum_{k = 1}^n (\sigma_{k,\textrm{approx}}^2 - \frac{1}{n}\sum_{l = 1}^{n}\sigma_{l,\textrm{approx}}^2)^2)$ - Empirical asymptotic variances, mean and variance over observable entries, where full gradients have been used.}
\label{table1}
\end{center}
\end{table}
\begin{table}[h!]
\begin{center}
\begin{tabular}{ | c | c | c | } 
\hline
 & block-size $T=300$  & block-size $T=9900$ \\ 
\hline
$\Gamma = I_n$ & (1.9575,0.0744) & (1.3338,1.6650)\\
\hline
$\Gamma = 0.2I_n$ & (0.4600,0.0042) & (0.2781,0.0784)\\
\hline
$\Gamma = 0.1I_n$ & (0.2646,0.0016) & (0.1335,0.0208)\\
\hline
overdamped & (1.9137,0.0791) & (1.3065,1.9714)\\
\hline
irreversible overdamped & (0.8764,0.0150) & (0.5778,0.3266)\\
\hline
\end{tabular}
\caption{The same as in Table~\ref{table1}, except for minibatch gradients}
\label{table2}
\end{center}
\end{table}

\begin{figure}[h]
  \centering
  \begin{subfigure}[h]{0.4\linewidth}
    \adjincludegraphics[trim={8mm 0 2mm 0},width=\linewidth]{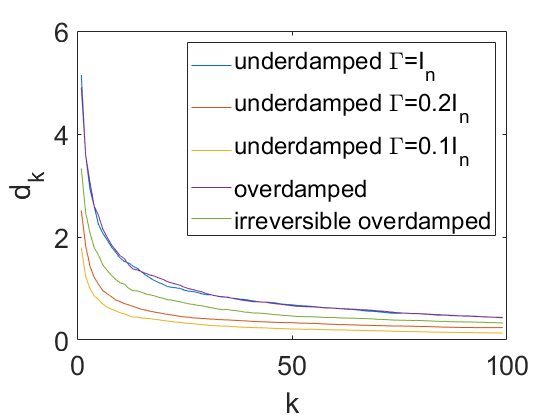}
  \end{subfigure}
  \begin{subfigure}[h]{0.4\linewidth}
    \adjincludegraphics[trim={8mm 0 2mm 0},width=\linewidth]{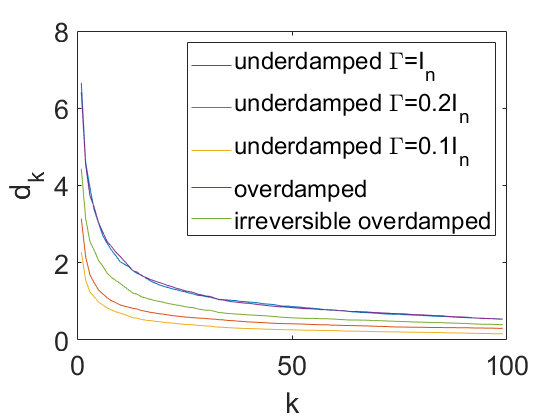}
  \end{subfigure}
  
  \caption{Euclidean distances to a combined posterior mean estimate over time. Left: full gradient. Right: minibatch gradient.}
  \label{intermeans}
\end{figure}

These figures demonstrate improvement of an order of magnitude in observed variances for $\Gamma$ close to that resulting from the gradient procedure over $\Gamma=I_n$. The improvement is also seen when compared to overdamped Langevin dynamics with and without irrreversible perturbation.

\section{Proofs}\label{proofs}

\begin{proof}\textit{(of Proposition~\ref{distprop})}
Take an approximating sequence $(f_k)_{k\in\mathbb{N}}$ such that $f_k\in C_c^\infty(\mathbb{R}^{2n})$, $f_k\rightarrow f$ in $L^2(\tilde{\pi})$. Fix $h\in C_c^\infty(\mathbb{R}^{2n})$ and for the differential operator $\mathcal{L}^*= -p^\top M^{-1}\nabla\!_q + \nabla U(q) ^\top \nabla\!_p - p^\top M^{-1}\Gamma \nabla\!_p + \nabla\!_p^\top \Gamma \nabla\!_p$ consider the expression
\begin{align}
\bigg|\int \mathcal{L}^*h \phi d\tilde{\pi}-\int h f d\tilde{\pi}\bigg| &\leq \abs*{\int \mathcal{L}^*h \phi d\tilde{\pi} - \int \mathcal{L}^*h \int_{\frac{1}{T}}^T P_t(f) dt d\tilde{\pi}}\nonumber\\
&\quad + \abs*{\int h(f-P_{\frac{1}{T}}(f) + P_T(f)) d\tilde{\pi}}\nonumber\\
&\quad + \abs*{\int \mathcal{L}^* h \int_{\frac{1}{T}}^T P_t(f-f_k)dtd\tilde{\pi}}\nonumber\\
&\quad + \abs*{\int h(P_{\frac{1}{T}}(f-f_k)-P_T(f-f_k))d\tilde{\pi}}\nonumber\\
&\quad + \abs*{\int \!\mathcal{L}^*h \int_{\frac{1}{T}}^T \! P_t(f_k) dt d\tilde{\pi}+\!\!\int \! h (P_{\frac{1}{T}}(f_k)-P_T(f_k)) d\tilde{\pi}}.\label{bigap}
\end{align}
For any $\epsilon>0$, $T$ can be chosen so that the first two terms on the right are each bounded by $\frac{\epsilon}{4}$ due to~\eqref{gTl2}, the strongly continuous property of $P_t$ and~\eqref{PTl2}; subsequently $k$ can be chosen so that the third and fourth terms are each bounded by $\frac{\epsilon}{4}$. For the remaining term, using Assumption~\ref{smu} for Theorem 5.13 in the first chapter of \cite{MR1730228} and H\"ormander's theorem \cite{MR222474}, 
\begin{equation}\label{backward}
\partial_t P_t (f_k) = \mathcal{L}P_t (f_k)
\end{equation}
holds classically on $(0,T)\times \mathbb{R}^{2n}$ and
\begin{equation}\label{smoothpt}
P_t(f_k)\in C^\infty ((0,T)\times\mathbb{R}^{2n}).
\end{equation}
By Fubini and equation~\eqref{backward},
\begin{align*}
\int \mathcal{L}^* h \int_{\frac{1}{T}}^T P_t(f_k) dt d\tilde{\pi} &= \int_{\frac{1}{T}}^T \int \mathcal{L}^* h P_t(f_k) d\tilde{\pi} dt \\
&= \int_{\frac{1}{T}}^T \int h\mathcal{L} P_t (f_k) d\tilde{\pi}dt\\
&= \int_{\frac{1}{T}}^T \int h \partial_t P_t (f_k) d\tilde{\pi}dt,
\end{align*}
so that since~\eqref{smoothpt} holds and therefore $P_t(f_k)$ is bounded on $(\frac{1}{T},T)\times B_R$ for any $k\in\mathbb{N}$, $T,R>0$, $B_R$ denoting the Euclidean ball in $\mathbb{R}^{2n}$ of radius $R$, Fubini's theorem can be applied again to obtain
\begin{equation*}
\int \mathcal{L}^* h \int_{\frac{1}{T}}^T P_t(f_k) dt d\tilde{\pi} = \int h (P_T (f_k) - P_{\frac{1}{T}}(f_k)) d\tilde{\pi},
\end{equation*}
which is that the last term in~\eqref{bigap} is equal to zero.

\end{proof}

Before giving the proof for the main formula of the directional derivative of the asymptotic variance, a truncation function is introduced and the membership of $\nabla_p \phi$ in $L^2(\tilde{\pi})$ is shown. The truncation function is constructed to satisfy a property~\eqref{etaprop} related to the generator~\eqref{infgen}; it will be used to robustly integrate by parts when establishing both $\nabla_p \phi \in L^2(\tilde{\pi})$ and the main formula.\\[0.5em]
Firstly, let $\varphi:\mathbb{R}\rightarrow\mathbb{R}$, $\varphi_k:\mathbb{R}\rightarrow\mathbb{R}$ be the standard mollifiers together with $\nu_k:\mathbb{R}\rightarrow\mathbb{R}$ be given by
\begin{equation*}
\varphi(x):=
\begin{cases}
e^{\frac{1}{x^2-1}}\bigg(\int_{-1}^1e^{\frac{1}{y^2-1}}dy\bigg)^{-1} &\qquad\textrm{if }-1<x\leq1\\
0 &\qquad\textrm{otherwise},
\end{cases}
\end{equation*}
\begin{equation}\label{mollialgo}
\varphi_k(x) := \frac{1}{k}\varphi\bigg(\frac{x}{k}\bigg),
\end{equation}
\begin{equation*}
\nu_k:=\varphi_k* \mathds{1}_{(-\infty,k^2]}\leq 1
\end{equation*}
for $k>0$.
\begin{lemma}\label{etaalgo}
Under Assumption~\ref{smu} and for $k>0$, let $\eta_{k}:\mathbb{R}^{2n}\rightarrow\mathbb{R}$ be the smooth functions given by
\begin{equation*}
\eta_{k}(\zeta)=\nu_k( \ln (1+\abs*{\zeta}^2))
\end{equation*}
for all $\zeta\in\mathbb{R}^{2n}$, then the following properties hold:
\begin{enumerate}
\item $\eta_k$ is compactly supported;
\item $\eta_k$ converges to 1 pointwise as $k\rightarrow\infty$;
\item for some constant $C>0$ independent of $k$,
\begin{equation}\label{etaprop}
\abs*{\mathcal{L}\eta_k} + (1+\abs*{\zeta})\abs*{\nabla \eta_k} \leq \frac{C}{k}.
\end{equation}
\end{enumerate}
\end{lemma}

\begin{proof}
The first two properties easily hold by definition of $\eta_k$. For the third property, denoting 
\begin{align*}
\mathcal{L}\eta_k &= \nu_k'(\ln(1+\abs*{\zeta}^2))\mathcal{L}\ln(1+\abs*{\zeta}^2)\\
&\quad+\nu_k''(\ln(1+\abs*{\zeta}^2))(\nabla\!_p\ln(1+\abs*{\zeta}^2))^\top \Gamma \nabla\!_p \ln(1+\abs*{\zeta}^2).
\end{align*}
It can be seen that $\nu_k'$ and $\nu_k''$ are estimated by terms at most of order $k^{-1}$; to see this, for all $x\in\mathbb{R}$,
\textcolor{black}{\begin{equation*}
\nu_k(x)=\int_{-\infty}^{k^2}\varphi_k(x-y)dy = \int_{x-k^2}^{\infty}\varphi_k(z)dz,
\end{equation*}
so that
\begin{equation*}
0 \geq \nu_k'(x) = -\varphi_k(x-k^2) \geq -k^{-1}\max \varphi
\end{equation*}
and
\begin{equation*}
\abs*{\nu_k''(x)} = \abs*{\varphi_k'(x-k^2)}\leq k^{-2}\max \varphi'.
\end{equation*}}
Therefore there exists a constant $\bar{C}>0$ such that
\begin{equation}\label{absl}
\abs*{\mathcal{L}\eta_k} \leq \bar{C}\Big(k^{-1}\abs*{\mathcal{L}\ln(1+\abs*{\zeta}^2)}+k^{-2}\abs*{(\nabla\ln(1+\abs*{\zeta}^2))^\top \Gamma \nabla\ln(1+\abs*{\zeta}^2)}\Big).
\end{equation}
Moreover, the expression
\begin{align*}
\mathcal{L}\ln(1+\abs*{\zeta}^2)&=\frac{\mathcal{L} \abs*{\zeta}^2}{1+\abs*{\zeta}^2}-\frac{4p^\top \Gamma p}{(1+\abs*{\zeta}^2)^2}\\
&=\frac{p^\top M^{-1}q - \nabla U(q)^\top p - p^\top M^{-1} \Gamma p + \textrm{Tr}\Gamma}{1+\abs*{\zeta}^2}-\frac{4p^\top \Gamma p}{(1+\abs*{\zeta}^2)^2}
\end{align*}
is bounded from above and below by~\eqref{assump2eq} and also the second term in~\eqref{absl} is bounded above by a direct calculation. Similarly, $(1+\abs*{\zeta})\abs*{\nabla \eta_k}$ can be bounded as required by a direct calculation.
\end{proof}

\begin{proof}\textit{(of Lemma~\ref{nabphi})}
Consider the functions $(f_{k,R})_{k,R\in\mathbb{N}}$ given by 
\begin{equation}\label{deffkr}
f_{k,R} := \varphi_{\frac{1}{k}}\ast (f\mathds{1}_{B_R})
\end{equation}
for $ \varphi_{\frac{1}{k}}$ given in~\eqref{mollialgo}, $B_R$ the radius $R$ ball in $\mathbb{R}^{2n}$ centered at $0$ and $\mathds{1}_{B_R}$ its indicator function. For any $k,R\in\mathbb{N}$, $f_{k,R}\in C_c^\infty$; moreover 
\begin{equation}\label{tendsl2}
f_{k_i,R_i}\rightarrow f
\end{equation}
in $L^2(\tilde{\pi})$ as $i\rightarrow\infty$ for any non-decreasing sequences $(k_i)_{i\in\mathbb{N}},(R_i)_{i\in\mathbb{N}}$ such that $k_i,R_i\rightarrow \infty$. By Theorem~\ref{poissonsolve} and Proposition~\ref{distprop} together with H\"ormander's theorem, there exists $\phi_{k,R}\in C^\infty \cap L_0^2(\tilde{\pi})$ such that 
\begin{equation*}
-\mathcal{L}\phi_{k,R} = f_{k,R} - \tilde{\pi}(f_{k,R})
\end{equation*}
for each $k,R\in\mathbb{N}$. For $r\in\mathbb{N}$, $\eta_r$ from Lemma~\ref{etaalgo} and the smallest eigenvalue $\lambda_m$ of $\Gamma$,
\begin{align}
\lambda_m\int \eta_r \abs*{\nabla\!_p \phi_{k,R}}^2 d\tilde{\pi} &\leq \int \eta_r \nabla\!_p \phi_{k,R} ^\top \Gamma \nabla\!_p \phi_{k,R} d\tilde{\pi}\nonumber\\
&= -\int \phi_{k,R} \nabla\!_p\eta_r^\top \Gamma \nabla\!_p \phi_{k,R} d\tilde{\pi} \nonumber\\
&\quad -\int \eta_r\phi_{k,R} (-M^{-1}p+\nabla\!_p)^\top(\Gamma \nabla\!_p \phi_{k,R}) d\tilde{\pi}. \label{ontheright}
\end{align}
The first term on the right can be written as
\begin{align*}
-\int \phi_{k,R} \nabla\!_p \eta_r^\top \Gamma \nabla\!_p \phi_{k,R} d\tilde{\pi} &= \int \phi_{k,R}(-M^{-1}p+\nabla\!_p)^\top\Gamma (\nabla\!_p\eta_r \phi_{k,R}) d\tilde{\pi}\\
&= \int \phi_{k,R}^2 (-M^{-1}p+\nabla\!_p)^\top \Gamma \nabla\!_p\eta_r d\tilde{\pi}\\
&\quad + \int \phi_{k,R} \nabla\!_p\eta_r^\top \Gamma \nabla\!_p \phi_{k,R} d\tilde{\pi},
\end{align*}
so that
\begin{align}
-\int \phi_{k,R} \nabla\!_p \eta_r^\top \Gamma \nabla\!_p \phi_{k,R} d\tilde{\pi} &= \frac{1}{2}\int\phi_{k,R}^2(-M^{-1}p+\nabla\!_p)^\top\Gamma \nabla\!_p \eta_r d\tilde{\pi}\nonumber\\
&= \frac{1}{2}\int\phi_{k,R}^2 \mathcal{L} \eta_r d\tilde{\pi}\nonumber\\
&\quad + \frac{1}{2}\int\phi_{k,R}^2 (-p^\top M^{-1} \nabla\!_q + \nabla U(q)^\top \nabla\!_p) \eta_r d\tilde{\pi}\nonumber\\
&\leq \frac{C}{r}\|\phi_{k,R}\|_{L^2(\tilde{\pi})}^2\label{stog}
\end{align}
for some generic constant $C>0$, where the last line follows by~\eqref{etaprop}. The second term on the right of~\eqref{ontheright} is
\begin{align}
-\int& \eta_r\phi_{k,R} (-M^{-1}p+\nabla\!_p)^\top(\Gamma \nabla\!_p \phi_{k,R}) d\tilde{\pi}\nonumber\\
&= -\int \eta_r\phi_{k,R} \mathcal{L}\phi_{k,R} d\tilde{\pi} + \int \eta_r\phi_{k,R} (p^\top M^{-1} \nabla\!_q - \nabla U(q)^\top \nabla\!_p) \phi_{k,R} d\tilde{\pi},\label{secondone}
\end{align}
where the last term integrates by parts to obtain
\begin{align*}
\int &\eta_r \phi_{k,R} (p^\top M^{-1} \nabla\!_q - \nabla U(q)^\top \nabla\!_p) \phi_{k,R} d\tilde{\pi}\\
&= -\frac{1}{2}\int (p^\top M^{-1} \nabla\!_q - \nabla U(q)^\top \nabla\!_p)\eta_r \phi_{k,R}^2 d\tilde{\pi},
\end{align*}
so that plugging back into~\eqref{secondone} and using again~\eqref{etaprop},
\begin{align*}
-\!\int\! \eta_r\phi_{k,R} (-M^{-1}p+\nabla\!_p)^\top(\Gamma \nabla\!_p \phi_{k,R}) d\tilde{\pi} &\leq \|\phi_{k,R} \|_{L^2(\tilde{\pi})} \| f_{k,R}-\tilde{\pi}(f_{k,R})\|_{L^2(\tilde{\pi})}\\
&\quad + \frac{C}{r}\|\phi_{k,R} \|_{L^2(\tilde{\pi})}^2.
\end{align*}
Plugging into~\eqref{ontheright}, together with~\eqref{stog},
\begin{equation*}
\lambda_m\int \eta_r\abs*{\nabla\!_p \phi_{k,R}}^2d\tilde{\pi} \leq \|\phi_{k,R} \|_{L^2(\tilde{\pi})} \| f_{k,R}-\tilde{\pi}(f_{k,R})\|_{L^2(\tilde{\pi})} + \frac{C}{r}\|\phi_{k,R}\|_{L^2(\tilde{\pi})}^2,
\end{equation*}
that is, $\nabla\!_p\phi_{k,R}\in L^2(\tilde{\pi})$. Since $\mathcal{L}$ and its parts are linear, the same arguments as above give 
\begin{align}
\lambda_m\int & \abs*{\nabla\!_p\phi_{k,R}-\nabla\!_p\phi_{k',R'}}^2d\tilde{\pi}\nonumber\\
& \leq \|\phi_{k,R} - \phi_{k',R'} \|_{L^2(\tilde{\pi})} \| f_{k,R} - f_{k',R'}-\tilde{\pi}(f_{k,R}) + \tilde{\pi}(f_{k',R'})\|_{L^2(\tilde{\pi})}\label{nabb}
\end{align}
for any $k',R'\in\mathbb{N}$. Moreover,
\begin{equation}\label{holdsalgo}
\abs*{f_{k,R}} = \abs*{\varphi_{\frac{1}{k}} \ast (f\mathds{1}_{B_R})} \leq C\mathds{1}_{B_{R+1}}\varphi_{\frac{1}{k}} \ast \mathcal{K}_l
\end{equation}
for a generic constant $C>0$ independent of $k,R$ and since $\varphi_{\frac{1}{k}}\ast \mathcal{K}_l\rightarrow \mathcal{K}_l$ uniformly on compact subsets, for any fixed $R\in\mathbb{N}$, there exists $K_R\in\mathbb{N}$ such that $k \geq K_R$ implies
\begin{equation}\label{klb}
\abs*{f_{k,R}} \leq 2C \mathcal{K}_l.
\end{equation}
Choosing now the sequences $R_i = i$, $k_i = i+ \max_{j\leq i} K_j$ for all $i\in\mathbb{N}$. By Theorems~\ref{thmc} and~\ref{poissonsolve}, $i,j\in\mathbb{N}$,
\begin{equation*}
\| \phi_{k_i,R_i} - \phi_{k_j,R_j} \|_{L^2(\tilde{\pi})} \leq C \left\| \frac{f_{k_i,R_i}-f_{k_j,R_j} - \tilde{\pi}(f_{k_i,R_i})+\tilde{\pi}(f_{k_j,R_j})}{\mathcal{K}_l} \right\|_{L^\infty},
\end{equation*}
where $C>0$ is again independent of $k,R$. Using the definition~\eqref{deffkr}, the terms $\abs*{\frac{\tilde{\pi}(f_{k,R})}{\mathcal{K}_l}}$ are bounded uniformly in $k$ and $R$, which, together with~\eqref{klb}, implies $\| \phi_{k_i,R_i} - \phi_{k_j,R_j} \|_{L^2(\tilde{\pi})} $ is bounded uniformly in $i,j$, so that inserting into~\eqref{nabb} gives
\begin{equation*}
\lambda_m\!\int \abs*{\nabla\!_p\phi_{k_i,R_i}-\nabla\!_p\phi_{k_j,R_j}}^2d\tilde{\pi}\leq C\| f_{k_i,R_i} - f_{k_j,R_j}-\tilde{\pi}(f_{k_i,R_i}) + \tilde{\pi}(f_{k_j,R_j})\|_{L^2(\tilde{\pi})}.
\end{equation*}
Together with~\eqref{tendsl2}, $\nabla\!_p\phi_{k_i,R_i}$ is a Cauchy sequence, with limit denoted as $g\in L^2(\tilde{\pi})$, so that for any $h\in C_c^\infty$,
\begin{equation*}
\abs*{\int g h + \int \phi \nabla\!_p h} \leq \abs*{\int g h - \int \nabla\!_p\phi_{k_i,R_i} h} + \abs*{\int \phi \nabla\!_p h  - \int \phi_{k_i,R_i} \nabla\!_p h},
\end{equation*}
hence 
\begin{equation}\label{caulim}
\nabla\!_p\phi_{k_i,R_i}\rightarrow g = \nabla\!_p \phi\in L^2(\tilde{\pi}).
\end{equation}
\end{proof}

Some additional preliminaries are presented here for the proof of Theorem~\ref{funcder}. For small $\epsilon\in\mathbb{R}$ and some direction $\delta\Gamma\in\mathbb{R}^{n\times n}$ in the space of smooth friction matrices such that $\Gamma + \epsilon\delta\Gamma\in \mathbb{S}_{++}^n$, let $L_{\epsilon}$ be the infinitesimal generator of~\eqref{langevin00} with the perturbed friction matrix $\Gamma+\epsilon\delta\Gamma$ in place of $\Gamma$, given formally by the differential operator
\begin{align*}
-\mathcal{L}_{\epsilon} &= -p^\top M^{-1} \nabla_q + \nabla U(q)^\top \nabla\!_p + p^\top M^{-1} (\Gamma+\epsilon\delta\Gamma) \nabla\!_p  - \nabla\!_p^\top (\Gamma+\epsilon\delta\Gamma)\nabla\!_p,
\end{align*}
where the notation $-\epsilon \mathcal{S}$ will be used for the perturbation on $\mathcal{L}$, that is,
\begin{equation*}
-\mathcal{S} = p^\top M^{-1} \delta\Gamma \nabla\!_p - \nabla\!_p^\top \delta\Gamma \nabla\!_p.
\end{equation*}
The formal $L^2(\tilde{\pi})$-adjoint of $\mathcal{L}_\epsilon$ is denoted
\begin{align*}
-\mathcal{L}_\epsilon^* &= p^\top M^{-1} \nabla_q - \nabla U(q)^\top \nabla\!_p + p^\top M^{-1} (\Gamma+\epsilon\delta\Gamma) \nabla\!_p - \nabla\!_p^\top (\Gamma+\epsilon\delta\Gamma)\nabla\!_p
\end{align*}
just as for $\mathcal{L}^*$.

\begin{proof}\textit{(of Theorem~\ref{funcder})}
For $\epsilon\leq \epsilon'$, by Theorem~\ref{poissonsolve} there exists a solution $\phi+\delta\phi_\epsilon\in L_0^2(\tilde{\pi})$ to the Poisson equation with the perturbed generator 
\begin{equation*}
-L_{\epsilon}(\phi+\delta\phi_\epsilon) = f - \pi(f).
\end{equation*}
By Theorem~\ref{theclt}, the directional derivative of $\sigma^2(\Gamma)$ in the direction $\delta\Gamma:\mathbb{R}^n \rightarrow \mathbb{R}^{n\times n}$ is
\begin{equation}\label{limitfunc}
\frac{1}{2}d\sigma^2.\delta\Gamma = \lim_{\epsilon\rightarrow 0}\frac{1}{\epsilon}\int \delta \phi_\epsilon f d\tilde{\pi}.
\end{equation}
Let $(f_{k,R})_{k,R\in\mathbb{N}}$ be given by~\eqref{deffkr}. Since inequality~\eqref{holdsalgo} holds by definition and $\varphi_{\frac{1}{k}}\ast g \rightarrow g$ uniformly on compact subsets for any continuous $g$, there exists for each $R\in\mathbb{N}$ a constant $\hat{K}_R\in\mathbb{N}$ such that $k\geq \hat{K}_R$ implies~\eqref{klb} for $C$ independent of $k,R$ and also 
\begin{equation*}
\abs*{f_{k,R}-f} \leq \frac{1}{R} \quad \textrm{on }B_{R-1}.
\end{equation*}
The sequences $R_i = i$, $k_i = i+\max_{j\leq i} \hat{K}_j$ for $i\in\mathbb{N}$ then give the sequence $(f_i)_{i\in\mathbb{N}}$, $f_i:=f_{k_i,R_i}\in C_c^\infty$ for all $i\in\mathbb{N}$, that satisfies
\begin{equation*}
\left\|\frac{f_i-f}{\mathcal{K}_{2l}}\right\|_{L^\infty} \rightarrow 0\qquad\textrm{as }i\rightarrow\infty,
\end{equation*}
which implies 
\begin{equation}\label{tend1}
\|f_i-f\|_{L^2(\tilde{\pi})}\rightarrow 0
\end{equation}
by~\eqref{fin} or by definition. Moreover, $\pi(f_i)\rightarrow 0$ by~\eqref{klb} and dominated convergence. Therefore the solutions $\phi_i,\phi_{i,\epsilon}\in L_0^2(\tilde{\pi})$ to the Poisson equations
\begin{equation*}
-L\phi_i = f_i - \pi(f_i) ,\qquad -L_{\epsilon}\phi_{i,\epsilon} = f_i - \pi(f_i)
\end{equation*}
given by Theorem~\ref{poissonsolve} satisfy
\begin{equation}\label{tend2}
\| \phi_i - \phi \|_{L^2(\tilde{\pi})} +  \| \phi_{i,\epsilon} - (\phi + \delta\phi_{\epsilon}) \|_{L^2(\tilde{\pi})} \rightarrow 0\quad \textrm{as } i\rightarrow \infty
\end{equation}
by~\eqref{poissol} and Theorem~\ref{thmc}. Since $f_i\in C^\infty$, H\"ormander's theorem together with Proposition~\ref{distprop} say that $\phi_i,\phi_{i,\epsilon} \in C^\infty$ and so $\phi_i,\phi_{i,\epsilon}$ solves $-\mathcal{L}\phi_i = -\mathcal{L}_{\epsilon}\phi_{i,\epsilon} = f_i-\pi(f_i)$ classically. Furthermore, in the same way as in the proof of Lemma~\ref{nabphi} to obtain~\eqref{caulim}, it holds that 
\begin{equation}\label{tend3}
\| \nabla\!_p\phi_i - \nabla\!_p\phi \|_{L^2(\tilde{\pi})} +\| \nabla\!_p\phi_{i,\epsilon} - \nabla\!_p(\phi + \delta\phi_\epsilon) \|_{L^2(\tilde{\pi})} \rightarrow 0\quad \textrm{as } i\rightarrow \infty,
\end{equation}
where $\nabla\!_p\phi_i,\nabla\!_p\phi_{i,\epsilon}\in L^2(\tilde{\pi})$ by Lemma~\ref{nabphi} itself.
The term under the limit in~\eqref{limitfunc} is now approximated with a term involving $f_i$ and the truncation functions $\eta_k$ from Lemma~\ref{etaalgo}.  Working now with the approximating integral and using Lemma~\ref{tildelem} together with the obvious extension on the notation from~\eqref{phitilde},
\begin{align}
\!\int\!\eta_k(\phi_{i,\epsilon}\! - \phi_i) f_i d\tilde{\pi} &= -\int\eta_k(\phi_{i,\epsilon}\! - \phi_i) \mathcal{L}_{\epsilon}^*\tilde{\phi}_{i,\epsilon} d\tilde{\pi}\nonumber\\
&=  -\int(\mathcal{L}_{\epsilon}\eta_k(\phi_{i,\epsilon}\! - \phi_i)+ 2\nabla\!_p\eta_k(\Gamma+\epsilon\delta\Gamma)\nabla\!_p(\phi_{i,\epsilon}\! - \phi_i)\nonumber\\
&\qquad \quad-\epsilon\eta_k\mathcal{S}\phi_i ) \tilde{\phi}_{i,\epsilon} d\tilde{\pi}\nonumber\\
&=  -\!\int\!(\mathcal{L}_{\epsilon}\eta_k(\phi_{i,\epsilon}\! - \phi_i)\!+ 2\nabla\!_p\eta_k(\Gamma+\epsilon\delta\Gamma)\nabla\!_p(\phi_{i,\epsilon}\! - \phi_i))\tilde{\phi}_{i,\epsilon}d\tilde{\pi}\nonumber\\
&\quad- \int\epsilon(\eta_k\nabla\!_p\phi_i^\top\delta\Gamma\nabla\!_p \tilde{\phi}_{i,\epsilon} + \tilde{\phi}_{i,\epsilon}\nabla\!_p\phi_i^\top\delta\Gamma\nabla\!_p \eta_k) d\tilde{\pi}.\label{dpf}
\end{align}
By Lemma~\ref{etaalgo} and~\ref{nabphi}, the terms involving gradients of $\eta_k$ converge to zero as $k\rightarrow \infty$, so that taking $k\rightarrow \infty$, then $i\rightarrow\infty$, 
\begin{equation*}
\int \delta\phi_\epsilon f d\tilde{\pi} = -\int\epsilon\nabla\!_p\phi^\top\delta\Gamma\nabla\!_p (\tilde{\phi} + \delta\tilde{\phi}_\epsilon) d\tilde{\pi}
\end{equation*}
holds, where~\eqref{tend1},~\eqref{tend2} and~\eqref{tend3} have been used. Plugging into~\eqref{limitfunc}, the directional derivative becomes
\begin{equation}\label{obl}
\frac{1}{2}d\sigma^2.\delta\Gamma = -\lim_{\epsilon \rightarrow 0} \int \nabla\!_p\phi^\top\delta\Gamma\nabla\!_p (\tilde{\phi} + \delta\tilde{\phi}_\epsilon) d\tilde{\pi}.
\end{equation}
From here, for any $\epsilon>0$, the unwanted term under the limit can be controlled by approximating again with a truncation and $f_i$, $i\in\mathbb{N}$,
\begin{align}
\lambda_m\int \eta_k&\abs*{\nabla\!_p (\tilde{\phi}_{i,\epsilon} - \tilde{\phi_i})}^2 d\tilde{\pi}\nonumber\\
 &\leq \int \eta_k \nabla\!_p (\tilde{\phi}_{i,\epsilon} - \tilde{\phi_i})^\top (\Gamma+\epsilon\delta\Gamma)\nabla\!_p (\tilde{\phi}_{i,\epsilon} - \tilde{\phi_i}) d\tilde{\pi}\nonumber\\
&= -\int  (\tilde{\phi}_{i,\epsilon} - \tilde{\phi_i})\nabla\!_p\eta_k^\top(\Gamma+\epsilon\delta\Gamma)\nabla\!_p(\tilde{\phi}_{i,\epsilon} - \tilde{\phi_i})d\tilde{\pi}\nonumber\\
&\quad - \int\eta_k (\tilde{\phi}_{i,\epsilon} - \tilde{\phi_i}) (-M^{-1}p + \nabla\!_p)^\top((\Gamma+\epsilon\delta\Gamma)\nabla\!_p(\tilde{\phi}_{i,\epsilon} - \tilde{\phi_i})d\tilde{\pi},\label{tog1e}
\end{align}
where 
\begin{equation*}
\lambda_m=\inf_{0<\epsilon\leq\epsilon'}\lambda_m^\epsilon
\end{equation*}
and $\lambda_m^\epsilon$ is the smallest eigenvalue of $\Gamma+\epsilon\delta\Gamma$. 
The first term on the right hand side is negligible as $k\rightarrow\infty$ because of Lemmata~\ref{etaalgo} and~\ref{nabphi}. The remaining term is 
\begin{align}
- \int\eta_k& (\tilde{\phi}_{i,\epsilon} - \tilde{\phi_i}) (-M^{-1}p + \nabla\!_p)^\top((\Gamma+\epsilon\delta\Gamma)\nabla\!_p(\tilde{\phi}_{i,\epsilon} - \tilde{\phi_i})d\tilde{\pi}\nonumber\\
 &= -\int \eta_k(\tilde{\phi}_{i,\epsilon} - \tilde{\phi_i})\mathcal{L}_\epsilon^* (\tilde{\phi}_{i,\epsilon} - \tilde{\phi_i})d\tilde{\pi}\nonumber\\
&\quad - \int \eta_k (\tilde{\phi}_{i,\epsilon} - \tilde{\phi_i})(p^\top M^{-1} \nabla\!_p - \nabla U^\top \nabla\!_p)(\tilde{\phi}_{i,\epsilon} - \tilde{\phi_i})d\tilde{\pi},\label{togg}
\end{align}
where the last term, after integrating by parts, gives
\begin{align}
\int \eta_k& (\tilde{\phi}_{i,\epsilon} - \tilde{\phi_i}) (p^\top M^{-1} \nabla\!_p - \nabla U^\top \nabla\!_p)(\tilde{\phi}_{i,\epsilon} - \tilde{\phi_i})d\tilde{\pi}\nonumber\\
&= -\frac{1}{2}\int (p^\top M^{-1} \nabla\!_q - \nabla U^\top \nabla\!_p)\eta_k (\tilde{\phi}_{i,\epsilon} - \tilde{\phi_i})^2d\tilde{\pi},\label{tog2e}
\end{align}
which is negligible as $k\rightarrow \infty$ again due to Lemma~\ref{etaalgo} and~\eqref{assump2eq}. On the other hand, the first term on the right hand side of~\eqref{togg} is 
\begin{align}
\!-\int \eta_k(\tilde{\phi}_{i,\epsilon} - \tilde{\phi_i}) \mathcal{L}_\epsilon^* (\tilde{\phi}_{i,\epsilon} - \tilde{\phi_i})d\tilde{\pi} &= \epsilon\int\eta_k(\tilde{\phi}_{i,\epsilon} - \tilde{\phi_i}) (M^{-1}p-\nabla\!_p)^\top\delta\Gamma \nabla\!_p\tilde{\phi}_i d\tilde{\pi}\nonumber\\
&= \epsilon\int\nabla\!_p\eta_k(\tilde{\phi}_{i,\epsilon} - \tilde{\phi_i})^\top\delta\Gamma\nabla\!_p\tilde{\phi}_id\tilde{\pi}\nonumber\\
&\quad +\epsilon\int\eta_k \nabla\!_p(\tilde{\phi}_{i,\epsilon} - \tilde{\phi_i})^\top\delta\Gamma\nabla\!_p\tilde{\phi}_i d\tilde{\pi},\label{tog3e}
\end{align}
where again the term involving $\nabla\!_p\eta_k$ is negligible as $k\rightarrow\infty$, so that putting together~\eqref{tog1e},~\eqref{togg},~\eqref{tog2e} and~\eqref{tog3e}, then taking $k\rightarrow \infty$ and $i\rightarrow \infty$ with~\eqref{tend3} gives
\begin{align*}
\lambda_m\int \abs*{\nabla\!_p\delta\tilde{\phi}_\epsilon}^2d\tilde{\pi}&\leq \epsilon\int \nabla\!_p\delta\tilde{\phi}_\epsilon^\top\delta\Gamma\nabla\!_p\tilde{\phi}d\tilde{\pi}\\
&\leq \frac{\epsilon\lambda_M}{\epsilon'}\int \bigg( \abs*{\nabla\!_p\delta\tilde{\phi}_\epsilon}^2 + \abs*{\nabla\!_p\tilde{\phi}}^2\bigg)d\tilde{\pi},
\end{align*}
where $\lambda_M=\inf_{0<\epsilon\leq\epsilon'}\lambda_M^\epsilon$, $\lambda_M^\epsilon$ is the largest eigenvalue of $\Gamma+\epsilon\delta\Gamma$ and $\delta\Gamma = \frac{\Gamma + \epsilon'\delta\Gamma - \Gamma}{\epsilon'}$ has been used. Therefore
\begin{equation*}
\int \abs*{\nabla\!_p\delta\tilde{\phi}_\epsilon}^2d\tilde{\pi} \leq \frac{\epsilon\lambda_M}{\epsilon'\lambda_m - \epsilon\lambda_M} \int \abs*{\nabla\!_p\tilde{\phi}}^2 d\tilde{\pi}
\end{equation*}
holds for small enough $\epsilon$ and putting into~\eqref{obl} concludes the proof.

\end{proof}

For the proof of Proposition~\ref{oddprop}, some notation is introduced. For $\tilde{k}\in\mathbb{N}$, let the tridiagonal matrix $M_{\tilde{k}}\in\mathbb{R}^{\tilde{k}+1 \times \tilde{k}+1}$ be given by its elements
\begin{equation}\label{Mkmat}
(M_{\tilde{k}})_{i,j} = \begin{cases}
i & \textrm{if } i + 1  = j,\\
(i-1)\gamma & \textrm{if } i = j,\\
i - \tilde{k} -2 & \textrm{if } i - 1 = j,\\
0 & \textrm{otherwise}
\end{cases}
\end{equation}
for indices $1\leq i,j \leq \tilde{k}+1$.

\begin{lemma}\label{Mklem}
Let $m\in\mathbb{N}$. Any tridiagonal matrix $\tilde{M}\in\mathbb{R}^{m\times m}$ of the form
\begin{equation*}
(\tilde{M})_{i,j} = \begin{cases}
b_i & \textrm{if } i + 1  = j,\\
b_i'\gamma & \textrm{if } i = j,\\
b_i'' & \textrm{if } i - 1 = j,\\
0 & \textrm{otherwise}
\end{cases}
\end{equation*}
for constants $b_i,b_i',b_i''\in\mathbb{R}$, has an order $\gamma$ determinant as $\gamma\rightarrow 0$ if $m$ is odd and a determinant that is bounded away from zero as $\gamma\rightarrow 0$ if $m$ is even.
\end{lemma}
Lemma~\ref{Mklem} is straightforwardly proved by repeatedly taking Laplace expansions. An explicit proof is not given here.

\begin{proof}\textit{(of Proposition~\ref{oddprop})}
Only a standard Gaussian and $M = 1$ is considered, the arguments for the general centered Gaussian case are the same. First consider the observable
\begin{equation}\label{polyob}
f(q) = q^k
\end{equation}
for some odd $k\in\mathbb{N}$. Take the polynomial ansatz
\begin{equation}\label{annie}
\phi(q,p) = \sum_{i,j = 0}^k a_{i,j} q^i p^j
\end{equation}
for $a_{i,j}\in\mathbb{R}$ and $\Gamma = \gamma >0$. It will be shown that arbitrarily small asymptotic variance is achieved in the $\gamma\rightarrow 0$ limit. Note that only pairs $(i,j)$ with odd $i$ and even $j$ make nonzero contributions to the asymptotic variance. Applying $-\mathcal{L}$ to the ansatz,
\begin{align*}
-\mathcal{L}\phi &= \sum_{i,j = 0}^k -ia_{i,j} q^{i-1} p^{j+1} + ja_{i,j}q^{i+1} p^{j-1} + \gamma j a_{i,j}q^i p^j - \gamma j (j-1) a_{i,j} q^i p^{j-2}\\
&= \sum_{i,j = 0} (-(i+1)a_{i+1,j-1} + (j+1)a_{i-1,j+1} + \gamma j a_{i,j} \\
&\qquad - \gamma (j+2)(j+1)a_{i,j+2} ) q^i p^j. 
\end{align*}
where
\begin{equation}\label{a1}
a_{i,j} = 0\quad \forall i,j <0 \textrm{ and } \forall i,j >k.
\end{equation}
Comparing coefficients in~\eqref{poisson0},
\begin{equation}\label{a2}
-(i+1)a_{i+1,j-1} + (j+1)a_{i-1,j+1} + \gamma j a_{i,j} - \gamma (j+2)(j+1)a_{i,j+2} = 0
\end{equation}
for all $(i,j)\neq (k,0)$. It holds by strong induction (in $j'$) that 
\begin{equation}\label{induct}
a_{i'+j',k+1-j'} = 0 \quad\forall i',j' \geq 0
\end{equation}
because of the following. The base case $j'=0$ follows by~\eqref{a1}, the induction step follows by taking $(i,j) = (i'+j'-1,k+2-j')$ for $i' \geq 0$ in~\eqref{a2} and again using~\eqref{a1} where necessary. Comparing coefficients in the Poisson equation~\eqref{poisson0} for $(i,j) = (k,0)$ and using~\eqref{a1},~\eqref{induct} yields\footnote{It is illustrative to imagine a grid of coefficients and the relations~\eqref{a2} and~\eqref{just1} as L-shaped chains on the grid, where~\eqref{induct} and~\eqref{a1} leave only a triangular area of nonzero coefficients.}
\begin{equation}\label{just1}
a_{k-1,1} = 1.
\end{equation}
Combining~\eqref{just1} with setting $(i,j) = (j'-1,k+1-j')$ for $j' = 1,\dots,k$ in~\eqref{a2}, the entries $a_{j',k-j'}$ satisfy the linear system
\begin{equation}\label{Mkeq}
M_k (a_{k,0},a_{k-1,1},\dots,a_{0,k})^\top = (1,0,\dots,0)^\top,
\end{equation}
where $M_k \in\mathbb{R}^{k+1\times k+1}$ is the tridiagonal matrix given in~\eqref{Mkmat}. In order to find the order in $\gamma$ as $\gamma\rightarrow 0$ of the elements of $(a_{k,0},\dots,a_{0,k})^\top$ appearing in~\eqref{Mkeq}, it suffices to find the order of the entries in the leftmost column of $M_k^{-1}$. 
For this, let $C_i\in\mathbb{R}$ be the $i^\textrm{th}$ minor appearing in the top row of the cofactor matrix of $M_k$. On the corresponding submatrix, repeatedly taking the Laplace expansion on the leftmost column until only the determinant of a $(k+1-i)$-by-$(k+1-i)$ square matrix from the bottom right corner of $M_k$ remains to be calculated, then using Lemma~\ref{Mklem} for this $(k+1-i)$-by-$(k+1-i)$ matrix 
gives that $C_i$ is of order $\gamma$ as $\gamma\rightarrow 0$ for odd $i$. Furthermore, the determinant of $M_k$ is bounded away from zero as $\gamma\rightarrow 0$ by Lemma~\ref{Mklem}. Therefore the elements of $(a_{k,0},\dots,a_{0,k})$ in the left hand side of~\eqref{Mkeq} with an odd index, that is $a_{k-j,j}$ for even $j$, have order $\gamma$ and at most order $1$ otherwise as $\gamma\rightarrow 0$. These elements with odd indices are exactly those from the vector $(a_{k,0},\dots,a_{0,k})^\top$ that make a contribution to the asymptotic variance. The `next' set of contributions come from the vector $(a_{k-2,0},a_{k-3,1}\dots,a_{0,k-2})$. 
Using again~\eqref{a1} and~\eqref{a2}, the vector satisfies
\begin{equation*}
M_{k-2} (a_{k-2,0},a_{k-3,1},\dots,a_{0,k-2})^\top = v_{k-2},
\end{equation*}
for some vector $v_{k-2}$ (from the last term on the left hand side of~\eqref{a2}) of order $\gamma$ as $\gamma\rightarrow 0$ and since the determinant of $M_{k-2}$ is of order $1$ (by Lemma~\ref{Mklem}), the contributions here to the asymptotic variance are again of order $\gamma$. Continuing for $(a_{k-2j,0},a_{k-2j-1,1}\dots,a_{0,k-2j})^\top$, $j\in\mathbb{N}$, it follows that all contributions are of order $\gamma$ as $\gamma\rightarrow 0$. The resulting coefficients indeed make up a solution $\phi$ to the Poisson equation because the matrices $M_k$ are invertible and because the coefficients $a_{i,j}$ for even $i+j$ are equal to zero from repeating the above procedure for the coefficients associated to $M_{k-1}$, $M_{k-3}$ and so on.  \\
For the general case of~\eqref{genodd}, since $\mathcal{L}$ is a linear differential operator and the contributions to the value of $\int\phi (f-\pi(f))d\tilde{\pi}$ come from exactly the same (odd $i$, even $j$) $a_{i,j}$ coefficients from the corresponding solution $\phi$ to each summand in~\eqref{genodd}, the proof concludes.
\end{proof}

\begin{proof}\textit{(of Proposition~\ref{quarprop})}
Take the polynomial ansatz
\begin{equation}\label{phan}
\phi(q,p) = \sum_{i,j = 0}^4 a_{i,j} q^i p^j
\end{equation}
for $a_{i,j}\in\mathbb{R}$, where $a_{i,j}$ not appearing in the sum are taken to be zero in the following. Again, only the standard Gaussian is considered, it turns out the arguments follow similarly otherwise. Comparing coefficients in~\eqref{poisson0} and using the same strong induction argument as in the proof of Proposition~\ref{oddprop} leads to~\eqref{a2} for all $(i,j)\neq(4,0),(0,0)$ and equation~\eqref{induct}. Taking $(i,j) = (j'-1, 5 - j')$ for $1\leq j' \leq 4$ in~\eqref{a2} and comparing the $q^4$ coefficients in the Poisson equation, it holds that
\begin{equation}\label{Mkeq3}
M_4 (a_{4,0},a_{3,1},a_{2,2},a_{1,3},a_{0,4})^\top = (1,0,\dots,0)^\top
\end{equation}
and taking $(i,j) = (j'-1, 3 - j')$ for $j' \geq 1$ in~\eqref{a2} yields
\begin{equation}\label{Mkeq4}
M_2 (a_{2,0},a_{1,1},a_{0,2})^\top = \gamma(2a_{2,2},6a_{1,3},12a_{0,4})^\top.
\end{equation}
Equations~\eqref{Mkeq3},~\eqref{Mkeq4} can be solved explicitly and the asymptotic variance is a weighted sum of the resulting coefficients. Those in~\eqref{phan} that make contributions are $a_{4,0},a_{2,2},a_{2,0}$, which gives the asymptotic variance $\frac{12(21\gamma^4 + 55\gamma^2 + 27)}{\gamma(3\gamma^2 + 4)}$ that goes to infinity as $\gamma\rightarrow 0$ or $\gamma\rightarrow\infty$. Comparing constant terms in the Poisson equation yields
\begin{equation*}
a_{0,2} = \frac{1}{2\Gamma}\int q^4 \frac{e^{\frac{q^2}{2}}}{\sqrt{2\pi}}dq = \frac{3}{2\Gamma},
\end{equation*}
which turns out to be satisfied by the solution for $a_{0,2}$, so that~\eqref{phan} is indeed a solution; note that the coefficients associated to $M_3$ and $M_1$ are zero by a similar procedure as above.
\end{proof}

\section{Discussion}\label{discussion}
\subsection{Relation to previous methodologies}
The infinite time integral~\eqref{pform} has been used for the calculation of transport coefficients in molecular dynamics \cite{MR3463433,MR2740040} and the derivative of the expectation appearing in~\eqref{pform} with respect to initial conditions is a problem considered when calculating the `greeks' in mathematical finance \cite{MR1842285}. On the topic of the latter and in contrast to \cite{MR1842285}, there is previous work dealing with cases of degenerate noise in the system, but the formulae derived were done so under different motivations and do not seem to improve upon~\eqref{nabphi0} in our situation; some of these references are given in Remark~\ref{alodir}.\\
Taking $\Gamma \rightarrow \infty$ together with a time rescaling, the dynamics~\eqref{langevin00} become the overdamped Langevin equation \cite{MR3288096}. An analogous result holds \cite{MR3324144} when $\Gamma = \Gamma(q)$ is position dependent
, where a preconditioner for the corresponding overdamped dynamics appears in terms of $\Gamma^{-1}$; see Section~\ref{posdep} for a consideration of our method in the position dependent friction case. 
On the other hand, 
the Hessian of $U$ makes a good preconditioner in the overdamped dynamics because of the Brascamp-Lieb inequality, see Remark 1 in \cite{MR3774648}.\\
On the application of underdamped Langevin dynamics with (variance reduced) stochastic gradients alongside the related Hamiltonian Monte Carlo method, \cite{pmlr-v139-zou21b} presents a comparison with convergence rates for the latter. In \cite{Chatterji2018OnTT}, convergence guarantees are provided for variance reduced gradients in the overdamped case and the control variate stochastic gradients in the underdamped case, along with numerical comparisons in low dimensional, tall dataset regimes. Furthermore, the underdamped dynamics with single, randomly selected component gradient update in place of the full gradient is considered in \cite{ding2020random}.\\
Variance reduction by modifying the observable instead of changing the dynamics has been considered for example in \cite{MR3969063, MR4108687,south2020regularised}. The methods there are incompatible with the framework in the present work due to the improved observable being unknown before the simulation of the Markov chain. Although useful, their applicability are limited in large $n$ cases due to storage requirements\cite{MR3969063}, not to mention either escalating computational cost for improvements in the observable or requirement of a priori knowledge\cite{south2020regularised}. 

\subsection{The nonconvex case}
In the case where $U$ is nonconvex, the Monte Carlo procedure in Section~\ref{thin} may continue to be used as presented, however the first variation process could easily stray from the case of exponential decay as in Theorem~\ref{newformprop}. Transitions from one metastable state to another cause the tangent process to increase in magnitude. In a one dimension double well potential $U(q) = \frac{q^4}{4} -q^2 + \frac{q}{2}$, linear observable $f(q) = q$ case, these transitions occur frequently enough during the gradient procedure in $\Gamma$ that $Dq$ blows up in simulation. Even in cases for which the metastabilities are strong, so that transitions occur less frequently, simulations show that $\Gamma$ dives to zero in periods where no transitions are occuring (as if the case of Corollary~\ref{lincor}), but increase dramatically in value once a transition does occur, causing the trajectory in $\Gamma$ to decay over time but occasionally jumping in value, so that there is no convergence for $\Gamma$. On the other hand, the Galerkin method presented in Appendix~\ref{num1} tends to give good convergence for $\Gamma$ in such cases.
\subsection{Position-dependent friction}\label{posdep}
It is possible to adapt the formula~\eqref{funcdereq} to the case of position-dependent gradient direction in $\Gamma$ given a Feynman-Kac representation formula and the corresponding existence result, which will be the aim of future work. The gradient direction is the same as~\eqref{DeltaGam} with the change that the integral is replaced by the corresponding marginal integral in $p$. 
Ideas using such a formula need to take into account that the first variation process retains a non-vanishing stochastic integral with respect to Brownian motion, 
so that the truncation in calculating the corresponding infinite time integral in Section~\ref{thin} is not as well justified, or rather, does not happen in the execution of Algorithm~\ref{algorithm} due to~\eqref{convconds} not being satisfied.
\subsection{Metropolisation}
Throughout Section~\ref{nummeth}, the implementation has not involved accept-reject steps. Metropolisation of discretisations of the underdamped Langevin dynamics was given in \cite{HOROWITZ1991247}, see also Section 2.2.3.2 in \cite{MR2681239} and \cite{MR4309974, doi:10.1063/1.2354490}. The systematic discretisation error is removed with the inclusion of this step but the momentum is reversed upon rejection (to avoid high rejection rates \cite{doi:10.1063/1.2354490}), which raises the question of whether friction matrices arising from Algorithm~\ref{algo0} improve the Metropolised situation where dynamics no longer imitate those in the continuous-time. For example the intuition in the Gaussian target measure, linear observable case discussed in Section~\ref{oddpol} no longer applies.
\subsection{Conclusion}
We have presented the central limit theorem for the underdamped Langevin dynamics and provided a formula for the directional derivative of the corresponding asymptotic variance with respect to a friction matrix $\Gamma$. A number of methods for approximating the gradient direction in $\Gamma$ have been discussed together with numerical results giving improved observed variances. Some cases where an improved friction matrix can be explicitly found have been given to guide the expectation of an optimal $\Gamma$. In particular, in cases where the observable is linear and the potential is close to quadratic, which is the case when finding the posterior mean in Bayesian inference with Gaussian priors, the optimal friction is expected to be close to zero (due to Corollary~\ref{lincor}). This is consistent with the numerical conclusion from the proposed Algorithm~\ref{algorithm}. Moreover, it is shown that the improvement in variance is retained when using minibatch stochastic gradients in a case of Bayesian inference.\\
We mention that the gradient procedure using~\eqref{funcder00} and~\eqref{nabphi0} can be used to guide $\Gamma$ in arbitrarily high dimension by extrapolation; that is, given a high dimensional problem of interest, the gradient procedure can be used on similar, intermediate dimensional problems in order to obtain a friction matrix that can be extrapolated to the original problem. In particular, for the Bayesian inference problem as formulated in Section~\ref{bayinf}, the algorithm recommends the choice of a small friction scalar, which can be expected to apply for datasets in an arbitrary number of dimensions. \\
Future directions not mentioned above includes well-posedness of the optimisation in $\Gamma$, extension to higher-order Langevin samplers methods as in \cite{chak2021generalised,MR4253735} and gradient formulae in the discrete time case analogous to Theorem~\ref{funcder}.

\section*{Acknowledgements}

M.C. was funded under a EPSRC studentship. G.A.P. was partially supported by
the EPSRC through grants EP/P031587/1, EP/L024926/1, and EP/L020564/1. N.K. and G.A.P. were funded in part by JPMorgan Chase $\&$ Co under a J.P. Morgan A.I. Research Awards 2019. Any views or
opinions expressed herein are solely those of the authors listed, and may differ from
the views and opinions expressed by JPMorgan Chase $\&$ Co. or its affiliates. This
material is not a product of the Research Department of J.P. Morgan Securities
LLC. This material does not constitute a solicitation or offer in any jurisdiction. Part of this project was carried out as T.L. was a visiting professor at Imperial College of London, with a visiting professorship grant from the Leverhulme
Trust. The Department of Mathematics at ICL and the Leverhulme Trust are warmly thanked for their support.

\bibliography{document}

\begin{appendices}

\section{Preliminaries}
\begin{theorem}\label{dent}
Let Assumption~\ref{smu} hold. For any $\mathcal{F}_0$-measurable $z_0:\Omega\rightarrow\mathbb{R}^{2n}$, there exists a unique almost surely continuous in $t$ solution $(q_t,p_t)=z_t:\Omega\rightarrow \mathbb{R}^{2n}$ to~\eqref{langevin00} that is $\mathcal{F}_t$-adapted and satisfies
\begin{equation}\label{membd}
\mathbb{E}\abs*{z_t}^2 \leq e^{\kappa t}(1+ \mathbb{E}\abs*{z_0}^2 )
\end{equation}
for a constant $\kappa\in\mathbb{R}$ and all $t\geq 0$. Furthermore, for any $z\in\mathbb{R}^{2n}$, $t\geq 0$, let $p_t^z$ be the probability measure given by
\begin{equation}\label{trans}
p_t^z(A) = \mathbb{P}(z_t^z\in A)
\end{equation}
for any Borel measurable $A$, where $z_t^z$ denotes the solution to~\eqref{langevin00} starting at $z_0 = z$, then $p_t^z$
\begin{enumerate}
\item is a transition probability in the sense that 
\begin{enumerate}
\item $(t,z)\mapsto p_t^z(A)$ is Borel measurable on $(0,\infty)\times \mathbb{R}^{2n}$,
\item the Chapman-Kolmogorov relation \cite{MR0494490} holds and
\end{enumerate}
\item \textcolor{black}{admits a density denoted $p(z,\cdot,t):\mathbb{R}^{2n}\rightarrow\mathbb{R}$ with respect to the Lebesgue measure on $\mathbb{R}^{2n}$ at every $(t,z)\in(0,\infty)\times\mathbb{R}^{2n}$ such that $p$ is a measurable function satisfying for every $z\in\mathbb{R}^{2n}$,
\begin{equation}\label{smoothden}
p(z,\cdot,\cdot)\in C^\infty(\mathbb{R}^{2n}\times(0,\infty)).
\end{equation}}
\end{enumerate}
\end{theorem}
\begin{proof}
Theorem 3.1.1 in \cite{MR2329435} yields existence and uniqueness of the solution to~\eqref{langevin00} together with~\eqref{membd}. Theorem 3.1 and 3.6 in Section 5 of \cite{MR0494490} give that $p_t^z(A)$ given by~\eqref{trans} is a probability kernel, that is, $p_t^z(A)$ is Borel measurable in $z$ for fixed $A,t$, is a probability measure in $A$ for fixed $z,t$ and satisfies the Chapman-Kolmogorov relation. For Borel measurability of $(t,z)\mapsto p_t^z(A)$ for fixed $A$, consider $\hat{z}_t^z$ given by
\begin{align}\label{modifi}
\hat{z}_t^z(\omega) = \begin{cases}
z_t^z(\omega) &\textrm{if }\omega : z_\bullet^z(\omega) \in C([0,\infty)),\\
0&\textrm{otherwise.}
\end{cases}
\end{align}
The process $\hat{z}_t^z$ is continuous in $t$ and $\mathcal{F}$-measurable in $\omega$, therefore 
$\mathbb{P}(\hat{z}_t^z\in A) = \mathbb{P}(z_t^z\in A)$ is continuous in $t$ hence Borel measurable in $(t,z)$. \textcolor{black}{Finally, $p_t^z$ admits a density at every $(t,z)\in(0,\infty)\times\mathbb{R}^{2n}$ satisfying~\eqref{smoothden} due to It\^{o}'s rule and H\"ormander's theorem \cite{MR222474}; measurability with respect to the starting point $z$ and therefore jointly in all of the arguments follows by the strong Feller property given by Theorem 4.2 in \cite{MR3256873}, because $p(t,\cdot,\zeta)$ is the pointwise limit of the continuous functions $(\int \eta_k(\zeta - \zeta') p(t,\cdot,\zeta')d\zeta')_{k>0}$, where $\eta_k$ denotes the standard scaled mollifiers.
}
\end{proof}
For all $t\geq 0$, all $z\in\mathbb{R}^{2n}$ and all $f:\mathbb{R}^{2n}\rightarrow\mathbb{R}$ integrable under the law $\mathcal{L}((z_t)_{t\geq 0}|z_0 = z)$ of $z_t$ starting at $z$, let 
\begin{equation}\label{semigroup}
P_t(f): z\mapsto \mathbb{E}(f(z_t^z)) = \mathbb{E}(f(z_t)|z_0 = z).
\end{equation}
The family $(P_t)_{t\geq 0}$ forms a strongly continuous (Proposition~\ref{strongcont}) Markov semigroup on $L^2(\tilde{\pi})$ with unit operator norm. Denote by $L$ the infinitesimal generator associated to this semigroup, given by
\begin{equation}\label{l2gen}
Lu = \lim_{t\rightarrow 0} \frac{P_t(u) - u}{t}
\end{equation}
for all functions $u\in\mathcal{D}(L)\subset L^2(\tilde{\pi})$, where the domain $\mathcal{D}(L)$ consists of the functions for which the above limit in $L^2(\tilde{\pi})$ exists.\\
\begin{prop}\label{strongcont}
The family $(P_t)_{t\geq 0}$ is strongly continuous in $L^2(\tilde{\pi})$.
\end{prop}
\begin{proof}
Fix $\epsilon>0$. For any $f\in L^2(\tilde{\pi})$, there exists $g\in C_c^\infty$ such that $\|f-g\|_{L^2(\tilde{\pi})}\leq \frac{\epsilon}{3}$. Writing
\begin{equation*}
\|P_tf - f\|_{L^2(\tilde{\pi})} \leq \|P_tf - P_tg\|_{L^2(\tilde{\pi})} + \| f - g \|_{L^2(\tilde{\pi})} + \| P_t g - g\|_{L^2(\tilde{\pi})},
\end{equation*}
the right hand side is bounded by $\epsilon$ after Jensen's inequality, invariance of $\tilde{\pi}$ and It\^o's rule together with a small enough $t$.
\end{proof}

\section{Solving the Poisson equation with a Galerkin method}\label{num1}
Throughout this appendix, $M=I_n$ is assumed. In low dimensions, 
it is feasible to approximate $\nabla\!_p\phi$ and a change in $\Gamma$ using Hermite polynomials. 
This approach gives an approximation in a finite subspace of $L^2(\tilde{\pi})$ at the level of $\phi$, as opposed to estimates of $\nabla\!_p\phi$ at particular points in space as in the Monte Carlo approach in Section~\ref{nummeth}. Specifically, the polynomials given by
\begin{equation*}
H_l(z) = \frac{(-1)^l}{\sqrt{l!}} e^{\frac{z^2}{2}} \frac{d^l}{dz^l}\bigg(e^{-\frac{z^2}{2}}\bigg)
\end{equation*}
for $l\in\mathbb{N}$ and their products in the multidimensional case
\begin{equation*}
H_{\underline{l}}(p) = \prod_{k = 1}^n H_{\underline{l}_{k}}(p_k), \qquad p=(p_1,\dots,p_n) \in\mathbb{R}^n
\end{equation*}
for multiindices $\underline{l}=(\underline{l}_1,\dots,\underline{l}_n)\in\mathbb{N}^n$ are considered in the weighted $L^2(\omega)$ space, where $\omega(p) = \frac{e^{-\frac{1}{2}\abs*{p}^2}}{(2\pi)^{-\frac{n}{2}}}$. A property of the Hermite polynomials that is repeatedly used here is that
\begin{equation*}
\partial_z H_l(z) = \sqrt{l} H_{l-1}(z).
\end{equation*}
For the application of Hermite polynomials in solving the Poisson equation associated to Langevin dynamics (in the case of scalar friction), we refer to \cite{MR3865558}. See also Chapter 5 in \cite{MR1641586} for Hermite polynomials in the multidimensional setting. In the case of a non-quadratic potential $U$, the same polynomials are used here after a Gram-Schmidt procedure in $L^2(\pi)$, which are denoted $(\hat{H}_{\underline{l}})_{\underline{l}\in\mathbb{N}^n}$, so that 
\begin{equation*}
\hat{H}_{\underline{l}} = \sum_{\abs*{\underline{k}}_\infty\leq K} \alpha_{\underline{l}}^{\underline{k}} H_{\underline{k}},
\end{equation*}
where $\abs*{\underline{k}}_\infty = \max(\underline{k}_1, \dots, \underline{k}_n)$, $K\in\mathbb{N}$, for some constants $\alpha_{\underline{k}}^{\underline{l}}\in\mathbb{R}$ calculated numerically. Their products with $H_{\underline{l}}$ are considered on $L^2(\tilde{\pi})$. Similarly, Fourier approximations can be used in the case of an $n$-torus (in $q$).\\
The observable $f\in L_0^2(\pi)$ is approximated by the projection defined by
\begin{equation}\label{fproj}
\Pi_K^q f := \sum_{\abs*{\underline{l}}_\infty\leq K} \hat{H}_{\underline{l}} \int f \hat{H}_{\underline{l}} d\pi = \sum_{\abs*{\underline{k}}_\infty,\abs*{\underline{l}}_\infty\leq K} \hat{H}_{\underline{l}} \alpha_{\underline{l}}^{\underline{k}} \int f H_{\underline{k}} d\pi.
\end{equation}
Since the generator has the form
\begin{equation*}
\mathcal{L} = \nabla_{\!p}^*\cdot \nabla_{\!q} - \nabla_{\!q}^* \cdot \nabla_{\!p} - (\nabla_{\!p}^*)^\top\Gamma\nabla_{\!p},
\end{equation*}
where
\begin{equation*}
\nabla_{\!q}^* = -\nabla_{\!q} + \nabla U,\qquad \nabla_{\!p}^* = -\nabla\!_p + p
\end{equation*}
are the respective formal $L^2(\tilde{\pi})$-adjoints of $\nabla\!_q$ and $\nabla\!_p$, the negative of the generator in the Poisson equation applied on functions of the form~\eqref{fproj} is the $(K+1)^{2n}$-by-$(K+1)^{2n}$ matrix given by 
\begin{align}
L_{\underline{k},\underline{l},\underline{\hat{k}},\underline{\hat{l}}} &= \langle \hat{H}_{\underline{k}} H_{\underline{l}}, -\mathcal{L} (\hat{H}_{\underline{\hat{k}}} H_{\underline{\hat{l}}})\rangle_{\tilde{\pi}}\nonumber\\
&= -\langle \hat{H}_{\underline{k}} \nabla\!_p H_{\underline{l}}, \nabla_{\!q}\hat{H}_{\underline{\hat{k}}} H_{\underline{\hat{l}}}\rangle_{\tilde{\pi}}  + \langle \nabla\!_q\hat{H}_{\underline{k}} H_{\underline{l}}, \hat{H}_{\underline{\hat{k}}} \nabla\!_p H_{\underline{\hat{l}}} \rangle_{\tilde{\pi}} + \langle \hat{H}_{\underline{k}} \nabla\!_p H_{\underline{l}},\Gamma\hat{H}_{\underline{\hat{k}}} \nabla\!_p H_{\underline{\hat{l}}}\rangle_{\tilde{\pi}}\nonumber\\
&= -\sum_i\langle\hat{H}_{\underline{k}},\partial_{q_i} \hat{H}_{\underline{\hat{k}}} \rangle_{\pi} (\sqrt{\underline{l}_i}\delta_{\underline{\hat{l}}}^{\underline{l}-e_i}) + \sum_i \langle\partial_{q_i}\hat{H}_{\underline{k}},\hat{H}_{\underline{\hat{k}}} \rangle_{\pi} (\sqrt{\underline{\hat{l}}_i}\delta_{\underline{\hat{l}}-e_i}^{\underline{l}})\nonumber\\
&\quad + \sum_{i,j}\delta_{\underline{\hat{k}}}^{\underline{k}} \delta_{\underline{\hat{l}}-e_i}^{\underline{l}-e_j} \sqrt{\underline{l}_j\underline{\hat{l}}_i} \Gamma_{i,j}\label{discgen}
\end{align}
where $\delta$ denotes the Kronecker delta here, the dependences of $\hat{H}_{\underline{k}}$, $\hat{H}_{\underline{\hat{k}}}$ and $H_{\underline{l}}$, $H_{\underline{\hat{l}}}$ on $q$ and $p$ respectively have been suppressed, $\langle v,w\rangle$ denotes $\sum_i\langle v_i, w_i \rangle$ for $v=(v_1,\dots,v_n)$, $w=(w_1,\dots,w_n)$ and $\langle \cdot,\cdot \rangle$ denotes the inner product on $L^2(\tilde{\pi})$. Note further that 
\begin{equation*}
\langle \partial_{q_i} \hat{H}_{\underline{k}}, \hat{H}_{\underline{\hat{k}}} \rangle_{\pi} = \sum_{\abs*{\underline{l}}_\infty,\abs*{\underline{\hat{l}}}_\infty\leq K}\alpha_{\underline{k}}^{\underline{l}}\sqrt{\underline{l}_i}\langle H_{\underline{l}-e_i},H_{\underline{\hat{l}}} \rangle_\pi \alpha_{\underline{\hat{k}}}^{\underline{\hat{l}}},
\end{equation*}
so that since $\alpha_{\underline{k}}^{\underline{l}}$ are derived from the inner products in $L^2(\pi)$ between the original Hermite polynomials $(H_{\underline{l}})_{\underline{l}}$, these inner products are the only values to be computed numerically other than those for the projection $\Pi_K^q f$ of the observable onto the finite dimensional subspace of $L^2(\tilde{\pi})$ spanned by the first $K+1$ Hermite polynomials given by~\eqref{fproj}. Solving the Poisson equation then reduces to finding the coefficients $\phi_{\underline{k},\underline{l}}\in\mathbb{R}$ of 
\begin{equation*}
\Pi_K^{(q,p)}\phi = \sum_{\abs*{\underline{k}}_\infty,\abs*{\underline{l}}_\infty\leq K}\phi_{\underline{k},\underline{l}} \hat{H}_{\underline{k}}H_{\underline{l}}
\end{equation*}
solving the linear system
\begin{equation}\label{linsys}
\sum_{\abs*{\underline{\hat{k}}}_\infty,\abs*{\underline{\hat{l}}}_\infty\leq K} L_{\underline{k},\underline{l},\underline{\hat{k}},\underline{\hat{l}}} \phi_{\underline{\hat{k}},\underline{\hat{l}}} = \Pi_K^{(q,p)} f = \begin{cases}
\sum_{\abs*{\underline{\hat{k}}}_\infty\leq K}\alpha_{\underline{k}}^{\underline{\hat{k}}}\int f H_{\underline{\hat{k}}} d\pi & \textrm{if } \underline{l} = \underline{0}\\
0 & \textrm{otherwise},
\end{cases}
\end{equation}
where note $L_{\underline{k},\underline{l},\underline{0},\underline{0}} = L_{\underline{0},\underline{0},\underline{\hat{k}},\underline{\hat{l}}} = 0$ so that only $\phi_{\underline{k},\underline{l}}$ for $(\underline{k},\underline{l})\neq (\underline{0},\underline{0})$ are determined by~\eqref{linsys} and $\phi_{\underline{0},\underline{0}}=0$ is enforced independently. Finally, the gradient direction in $\Gamma$ is given by
\begin{align}
(\Delta\Gamma)_{i,j} &= \int \sum_{\abs*{\underline{k}}_\infty,\abs*{\underline{l}}_\infty\leq K}\phi_{\underline{k},\underline{l}} \hat{H}_{\underline{k}}\sqrt{\underline{l}_i}H_{\underline{l}-e_i} \sum_{\abs*{\underline{\hat{k}}}_\infty,\abs*{\underline{\hat{l}}}_\infty\leq K} \phi_{\underline{\hat{k}},\underline{\hat{l}}} (-1)^{\abs*{\underline{\hat{l}}}}\hat{H}_{\underline{\hat{k}}}\sqrt{\underline{\hat{l}}_j}H_{\underline{\hat{l}}-e_j}d\tilde{\pi}\nonumber\\
&= \sum_{\abs*{\underline{k}}_\infty,\abs*{\underline{l}}_\infty,\abs*{\underline{\hat{k}}}_\infty,\abs*{\underline{\hat{l}}}_\infty\leq K}\phi_{\underline{k},\underline{l}}\phi_{\underline{\hat{k}},\underline{\hat{l}}}\sqrt{\underline{l}_i \underline{\hat{l}}_j }(-1)^{\abs*{\underline{\hat{l}}}} \delta_{\underline{\hat{k}}}^{\underline{k}} \delta_{\underline{\hat{l}}-e_j}^{\underline{l}-e_i}\nonumber\\
&= \sum_{\abs*{\underline{k}}_\infty,\abs*{\underline{l}}_\infty\leq K}\phi_{\underline{k},\underline{l}}\phi_{\underline{k},\underline{l}-e_i+e_j}\sqrt{\underline{l}_i (\underline{l}-e_i+e_j)_j }(-1)^{\abs*{\underline{l}}}\label{delgam1}
\end{align}
where $\abs*{\underline{\hat{l}}} = \underline{\hat{l}}_1+\dots+\underline{\hat{l}}_n$ and $\phi_{\underline{k},\underline{l}}=0$ if there is some $i$ such that $\underline{k}_i>K$ or $\underline{l}_i>K$. More robustly, the asymptotic variance can be discretised first, followed by taking the gradient direction with respect to the approximate asymptotic variance. Namely, (half of) the asymptotic variance $\int \nabla\!_p \phi^\top \Gamma \nabla\!_p\phi d\tilde{\pi}$ can be approximated by
\begin{equation}\label{approxav}
\sum_{\abs*{\underline{k}}_\infty,\abs*{\underline{l}}_\infty,\abs*{\underline{\hat{k}}}_\infty,\abs*{\underline{\hat{l}}}_\infty\leq K}\phi_{\underline{k},\underline{l}}L_{\underline{k},\underline{l},\underline{\hat{k}},\underline{\hat{l}}}\phi_{\underline{\hat{k}},\underline{\hat{l}}}
\end{equation}
(or simply the last term in~\eqref{discgen} replacing $L_{\underline{k},\underline{l},\underline{\hat{k}},\underline{\hat{l}}}$), so that the derivative with respect to the entries $\Gamma_{i,j}$ can be taken as follows. With abuse of notation, let $L^{-1}\in\mathbb{R}^{(K+1)^{2n}-1\times (K+1)^{2n}-1}$ be the inverse of the matrix depending on $\Gamma$ given by~\eqref{discgen} with the $L_{\underline{k},\underline{l},\underline{0},\underline{0}} = L_{\underline{0},\underline{0},\underline{\hat{k}},\underline{\hat{l}}} = 0$ entries removed. Let also $\underline{\phi}\in\mathbb{R}^{(K+1)^{2n}}$ be the vector made up of the coefficients $\phi_{\underline{k},\underline{l}}$ for $\underline{k}+\underline{l}\neq 0$ so that equation~\eqref{linsys} can be rewritten as 
\begin{equation*}
\underline{\phi}= L^{-1} ((\Pi_K^{(q,p)}f)_2,\dots,(\Pi_K^{(q,p)}f)_{(K+1)^{2n}})^\top.
\end{equation*}
By~\eqref{discgen}, the derivative of $L_{\underline{k},\underline{l},\underline{\hat{k}},\underline{\hat{l}}}$ with respect to the entry $\Gamma_{i,j}$ is 
\begin{equation*}
\partial L_{\underline{k},\underline{l},\underline{\hat{k}},\underline{\hat{l}}}^{i,j}:=\delta_{\underline{\hat{k}}}^{\underline{k}} \delta_{\underline{\hat{l}}-e_i}^{\underline{l}-e_j} \sqrt{\underline{l}_j\underline{\hat{l}}_i}.
\end{equation*}
Let $\partial L^{i,j}\in\mathbb{R}^{(K+1)^{2n}-1\times(K+1)^{2n}-1}$ denote the matrix with entries $\partial L_{\underline{k},\underline{l},\underline{\hat{k}},\underline{\hat{l}}}^{i,j}$ except the $\partial L_{\underline{k},\underline{l},\underline{0},\underline{0}}^{i,j}$ and $\partial L_{\underline{0},\underline{0},\underline{\hat{k}},\underline{\hat{l}}}^{i,j}$ entries are deleted. The derivative of~\eqref{approxav} with respect to the entry $\Gamma_{i,j}$ is then
\begin{equation}\label{discthengrad}
\underline{\phi}^\top \partial L^{i,j} \underline{\phi} + \sum_{\abs*{\underline{k}}_\infty,\abs*{\underline{l}}_\infty,\abs*{\underline{\hat{k}}}_\infty,\abs*{\underline{\hat{l}}}_\infty\leq K} 2(\partial\phi^{i,j})_{\underline{k},\underline{l}}L_{\underline{k},\underline{l},\underline{\hat{k}},\underline{\hat{l}}}\phi_{\underline{\hat{k}},\underline{\hat{l}}},
\end{equation}
where $\partial \phi^{i,j}\in\mathbb{R}^{(K+1)^{2n}}$ is the vector given by
\begin{equation*}
(\partial \phi^{i,j})_k := \begin{cases}
0 & \textrm{if } k = 1\\
-(L^{-1}\partial L^{i,j} \underline{\phi})_{k-1}& \textrm{otherwise},
\end{cases}
\end{equation*}
so that the gradient direction in $\Gamma$ is given by the negative of~\eqref{discthengrad}.

It's also possible to approximate $\phi$ using a finite difference in $q$, Hermite projection in $p$ approach in the case when the state space in $q$ is the $n$-torus; we omit further descriptions but refer to \cite{MR1933042} for this direction.

\section{Approximation of $\Delta\Gamma$ using independent realisations}\label{cigam1}
One can use the endpoints of a number of independent realisations of~\eqref{langevindisc} to approximate the integral with respect to $\pi$ in~\eqref{expecs} and, for each of those realisations, to use two additional sets of realisations of~\eqref{langevindisc} and~\eqref{Dqdisc} to approximate each of the expectations under the integral in~\eqref{expecs}.\\
Fix a starting point $(q,p)$; the first of the expectations in~\eqref{expecs} (and similarly for the second) can be approximated at time $s=i\Delta t$ with
\begin{equation*}
\frac{1}{K}\sum_{k=1}^K \nabla f(q_{(k)}^i)^\top Dq_{(k)}^i
\end{equation*}
where $K\in\mathbb{N}$, $(q_{(k)}^i,p_{(k)}^i)_{i\in\mathbb{N}}$ denotes the solution to~\eqref{langevindisc} with initial condition $(q,p)$, noise $\xi^i = \xi_{(k)}^i$ for all $i\in\mathbb{N}$ and where $\xi_{(k)}^i$ are independent as $k=1,\dots,K$, $i$ varies and $(Dq_{(k)}^i,Dp_{(k)}^i)$ is the corresponding solution to~\eqref{Dqdisc}. Subsequently, introducing an additional population of independent realisations of~\eqref{langevindisc} to draw from $\pi$ after some burn-in period, the change~\eqref{expecs} in $\Gamma$ can be approximated by
\begin{equation*}
-\frac{1}{L}\sum_{l=1}^L \bigg( \sum_{i=B+1}^{B+T}  \frac{\Delta t}{K}\sum_{k = 1}^K \nabla f(q_{(l,k)}^i)^\top Dq_{(l,k)}^i \bigg)^{\!\!\!\top}\! \bigg(  \sum_{i=B+1}^{B+T} \frac{\Delta t}{K}\sum_{k = 1}^K \nabla f(\tilde{q}_{(l,k)}^i)^\top D\tilde{q}_{(l,k)}^i \bigg)
\end{equation*}
where $B\in\mathbb{N}$ is some burn-in number of iterations, $T\in\mathbb{N}$ is some a posteriori number of iterations depending on whether the magnitude of the entries of $Dq_{(l,k)}^{B+T}$ are smaller than some fixed value for all $k,l$; furthermore $L\in\mathbb{N}$, $((q_{(l,k)}^i,p_{(l,k)}^i))_{i\in\mathbb{N}}$ denotes the solution to~\eqref{langevindisc} with initial condition say $(0,0)$, noise $\xi^i = \xi_{(l,k)}^i$ for all $i\in\mathbb{N}$ satisfying
\begin{equation*}
\xi_{(l,k)}^i = \xi_{(l,k')}^i \qquad \forall i < B, 1\leq k,k' \leq K
\end{equation*}
and are independent otherwise, $((\tilde{q}_{(l,k)}^i,\tilde{p}_{(l,k)}^i))_{i\geq B}$ denotes the solution to~\eqref{langevindisc} with `initial' condition
\begin{equation*}
(\tilde{q}_{(l,k)}^B,\tilde{p}_{(l,k)}^B) = (q_{(l,k)}^B,-p_{(l,k)}^B)
\end{equation*}
for all $1\leq k\leq K$, $1\leq l \leq L$, independent noise $\xi^i = \tilde{\xi}_{(l,k)}^i$ for $i\geq B$ independent also to $(\xi_{(l,k)}^i)_{i\in\mathbb{N}}$. The notation $(Dq_{(l,k)}^i,Dp_{(l,k)}^i)$, $(D\tilde{q}_{(l,k)}^i,D\tilde{p}_{(l,k)}^i)$ represent the corresponding solutions to~\eqref{Dqdisc}.
\end{appendices}

\end{document}